\documentclass{article}
\usepackage{amsmath,amsthm,amsfonts,graphicx}
\usepackage{multirow}
\usepackage{slashbox}
\usepackage{multirow}
\def\smallddots{\mathinner{\raise7pt\hbox{.}\raise4pt\hbox{.}\raise1pt\hbox{.}}}
\def\smallsdots{\mathinner{\raise1pt\hbox{.}\raise4pt\hbox{.}\raise7pt\hbox{.}}}

\DeclareMathOperator{\diag}{diag}

\DeclareMathOperator{\rank}{rank}

\numberwithin{equation}{section}
\numberwithin{table}{section}
\newtheorem{theorem}{Theorem}[section]
\newtheorem{lemma}{Lemma}[section]

\newtheorem{corollary}{Corollary}[section]

\newtheorem{definition}{Definition}[section]

\newtheorem{remark}{Remark}[section]

\begin{document}

\title{\bf Numerically Safe Gaussian Elimination  \\
with No Pivoting
 \thanks {Some results of this paper have been presented at the 
17th Annual Conference on Computer Algebra in Scientific Computing (CASC'2014), September 10--14, 2015, Aachen, Germany (cf. \cite{PZ15}).}}

\author{Victor Y. Pan$^{[1, 2],[a]}$ and Liang Zhao$^{[2],[b]}$
\and\\
$^{[1]}$ Department of Mathematics and Computer Science \\
Lehman College of the City University of New York \\
Bronx, NY 10468 USA \\
$^{[2]}$ Ph.D. Programs in Mathematics  and Computer Science \\
The Graduate Center of the City University of New York \\
New York, NY 10036 USA \\
$^{[a]}$ victor.pan@lehman.cuny.edu \\
http://comet.lehman.cuny.edu/vpan/  \\
$^{[b]}$  lzhao1@gc.cuny.edu \\
} 
 \date{}

\maketitle


\begin{abstract}
{\em Gaussian elimination with no pivoting} and {\em block Gaussian elimination} are
attractive alternatives to the customary but communication intensive {\em Gaussian elimination with partial
 pivoting}\footnote{Hereafter we use the acronyms {\em GENP},  {\em BGE}, and {\em GEPP}.} provided that the computations proceed {\em safely} and {\em numerically safely}, that is,  run into neither division by 0 nor numerical problems. Empirically, safety and numerical safety of GENP  have been consistently observed in a number of papers  where an input matrix was pre-processed with various structured multipliers
chosen  ad hoc. Our present paper  provides missing formal support for this empirical observation and explains why it was elusive so far.  Namely we prove that  GENP is  numerically unsafe for  a specific  class of input matrices in spite of its pre-processing with some well-known 
and well-tested structured multipliers,
but we also  prove that GENP and BGE are safe and numerically safe for the average input matrix  pre-processed with any nonsingular and well-conditioned multiplier.
 This should
embolden search for sparse and structured
multipliers, 
and we list and test some new classes of them.
 We also seek randomized pre-processing  that universally (that is, for all input matrices) supports (i) safe GENP and BGE with  probability 1  and/or (ii) numerically safe GENP and BGE with a  probability close to 1.
We achieve goal (i) with a Gaussian
structured multiplier and goal (ii) with a Gaussian unstructured multiplier and alternatively  with Gaussian structured augmentation.  We consistently confirm all these formal  results with our tests of GENP for benchmark inputs. We have extended our approach to other fundamental matrix computations
and keep working on further extensions.
  \end{abstract}






\section{Introduction}\label{sintr}


\subsection{Gaussian elimination with pivoting}\label{sgep}


The history of 
Gaussian elimination can be traced back some 2000 years \cite{G11}.
Its modern version,
GEPP (that is, Gaussian elimination with partial pivoting),
is performed routinely, millions times per day 
around the world, being a cornerstone for 
computations in linear algebra \cite{DDF14}. 
For an $n\times n$  matrix, elimination involves
about $\frac{2}{3}n^3$ 
flops, and $(n-1)n/2$ 
comparison are required  for partial pivoting,
that is, row interchange.\footnote{Here and hereafter ``flop" stands for
``floating point arithmetic operation".}
Clearly pivoting contributes a substantial share to the 
overall computational cost if $n$ is 
small, but   
 communication intensive pivoting
takes quite a heavy toll in modern computer environment even for larger $n$. Pivoting
interrupts the stream of arithmetic operations 
with foreign operations of comparison,
involves book-keeping, compromises data locality,
impedes parallelization of the computations, and
increases communication overhead and data dependence.
According to \cite{BDHT13},
``pivoting can represent more than 40\% of the global factorization time for
small matrices, and although the overhead decreases with the size of the matrix, it still
represents 17\% for a matrix of size 10,000".
Because of the heavy use of GEPP,  even its limited improvement 
is valuable.


 \subsection{Contribution: random  and nonrandom multipliers, safety and numerical safety}\label{smvp}


{\em Gaussian elimination with no pivoting}
(GENP) is an attractive alternative to  GEPP\footnote{Another alternative is symmetrization, but it has deficiencies for both numerical and symbolic computations: it squares the condition number of the input matrix
and does not work over finite fields.}, 
but for some inputs  can be {\em unsafe} or  {\em numerically unsafe}, that is,  can run into 
a division by 0 or numerical problems, respectively. 
  Empirically, GENP is quite consistently  safe and numerically safe   \cite{PQZ13}, \cite{BDHT13}, \cite{PQY15}, \cite{DDF14}  if the input matrix is pre-processed  with various structured multipliers chosen ad hoc  
 (e.g., with random circulant   or  SRTF   
 multipliers),\footnote{SRFT is the acronym for Semisample Random Fourier Transform} but formal support for this empirical observation turned out to be elusive. 
 
We explain why it is elusive  by exhibiting a class of input matrices for which GENP with a random circulant or SRFT multiplier is always numerically unsafe,
but we also explain why pre-processing with  
such multipliers  usually works -- we prove that {\em GENP is both safe and numerically safe  for the average input matrix pre-processed with any nonsingular and well-conditioned multiplier}
provided that  the average is defined  under the Gaussian probability distribution,
which is a customary assumption in view of the Central Limit Theorem. 

 We specify and 
successfully  test some promising classes of such sparse and structured
multipliers, and plan to work further in this direction.
For  inputs to our tests we used 
 {\em benchmark matrices} from \cite{BDHT13} 
  or applied 
 the {\em customary recipes} of
  \cite[Section 28.3]{H02}.


 \subsection{Contribution: universal pre-processing}\label{sunvp}


In addition to our results on pre-processing GENP with a fixed multiplier for the average input matrix,
we prove that pre- as well as post-multiplication by a Gaussian 
multiplier\footnote{Here and hereafter ``Gaussian matrices" and ``Gaussian pre-processing" stand for ``standard Gaussian random matrices" and  ``standard Gaussian random pre-processing", respectively.} are {\em universally safe}, that is, safe with probability 1 for any fixed nonsingular input matrix, and is {\em universally numerically safe}, that is, likely to be numerically safe for  any nonsingular and well-conditioned input matrix.

In computations with infinite precision and with no rounding errors (e.g., in Computer Algebra) one needs just safe  (rather than numerically
safe)  GENP, and  we prove universal safety  of GENP
  even in the case of pre- or post-multiplication by a random Gaussian circulant matrix.
This result immediately  enables 4-fold acceleration of the
classical, extensively cited, and highly popular two-sided pre-processing  of \cite{KS91}.

We cannot prove universal numerical safety of GENP pre-processed with a  random  multiplier from any class of structured matrices,
but we prove that  GENP 
pre-processed with 
{\em SRFT augmentation} or {\em SRFT additive preprocessing}
is universally safe with probability 1 and is   likely to be universally
numerically safe.


 \subsection{Contribution: numerical tests}\label{sntsts}


The results of our extensive numerical tests of GENP are in very good accordance with our formal analysis. This was also the case in the previous studies of GENP 
in   \cite{PQZ13}, \cite{BDHT13}, \cite{PQY15}, and \cite{DDF14}, but
 our present tests cover pre-processing also with multipliers from  new classes as well as by means of augmentation. 
 

 \subsection{Contribution: block Gaussian elimination}\label{sbge}


Block matrix algorithms are highly important 
resource for enhancing the efficiency of matrix computations \cite{GL13}, but {\em block Gaussian elimination} (BGE),
including,
e.g.,   the MBA celebrated
superfast algorithm
for solving structured linear systems of equations, is prone to numerical stability problems (cf. \cite{B85}, \cite[Chapter 5]{P01}). 
We readily extend our results for GENP to BGE.\footnote{By combining our dual approach 
with the study in \cite{YC97} of the growth factor in PLUP' factorization of the input matrix,
we can deduce the results  similar to our present ones  for GENP but not for BGE.}
In particular a proper random structured pre-processing is  likely to make BGE   safe and numerically safe. 
This enables us to resurrect the MBA algorithm for  numerical computations.
Furthermore we  accelerate our pre-processing for BGE  by performing it recursively for leading block submatrices rather than once 
 for the whole matrix. 
  

 \subsection{Organization of the paper}\label{sorgp}


We organize 
our presentation as follows.
In the next section   we cover basic definitions
and some preliminary results.
In Section \ref{sbgegenpn} we define BGE and show that,
for any input, safety as well as numerical safety of  GENP imply the same properties also for BGE.
In Section  
\ref{s0}  we prove our basic theorems 
about safety
and numerical safety
of  GENP (and consequently BGE) with multiplicative pre-processing and at the end cover some techniques for enhancing the power and  efficiency 
of pre-processing, in particular by means of recursive
  pre-processing, applied successively  to the current leading blocks rather than once to the whole input matrix
(see Remark \ref{reblock}).
In Section \ref{sbstr}
we recall the definitions and some basic properties of the matrices of Discrete
Fourier transform and circulant and factor-circulant ($f$-circulant) matrices.
In Section  \ref{ssprsml} we describe some families of efficient multipliers.
In  Section \ref{ssmblnm1s} we prove that 
randomized circulant pre-processing for GENP and BGE  is universally safe 
(the contribution
of the second author) but is numerically unsafe for some specific input class
of matrices. 
In Section \ref{saltaug} we 
prove universal safety of  augmentation and additive pre-processing with Gaussian as well as SRFT pre-processors. Section \ref{sexp1} (also the contribution
of the second author) covers 
our numerical experiments with benchmark inputs and inputs generated according to 
\cite[Section 18.3]{H02}, although performance (runtime and Gigaflop/s)
on current parallel machines has not been addressed in this paper and is postponed for future work.
 In Section \ref{sconc} we briefly recall our progress and then comment on its extension  to low-rank approximation and
 other fundamental matrix computations.
A number of our results 
for Gaussian  pre-processing
can be extended to pre-processing under  the uniform 
probability distribution on a finite set (cf. Theorem \ref{thrnd}).


\section{Basic definitions and preliminary results}\label{sdef}


Hereafter ``likely" 
and ``unlikely" mean ``with a probability close to 1" 
and, respectively, ``to 0". 
 We call an $m\times n$ matrix {\em  Gaussian} and denote it $G_{m,n}$ 
if all its entries are i.i.d. standard Gaussian variables. 
$\mathcal G^{m\times n}$, $\mathbb R^{m\times n}$, and $\mathbb C^{m\times n}$
 denote the classes of $m\times n$ Gaussian, real and complex matrices, respectively.
In order to simplify our presentation, we assume dealing with real matrices,
except for the matrices of discrete Fourier transform of Section \ref{stpldft}
and $f$-circulant matrices of  Section \ref{stplcrc} used in Sections \ref{ssprsml} and \ref{sexp1},
but we can  readily extend our study to the complex case. 

  

\subsection{General matrices: basic definitions}\label{sgnm}
In this subsection we recall some basic definitions for general matrix  computations (cf. \cite{GL13}).   
\begin{enumerate} 
\item
$I_g$ is a $g\times g$ identity matrix.
$O_{k,l}$ is the  $k\times l$ matrix filled with zeros.
\item
  $(B_1~|~B_2~|~\cdots~|~B_k)$ is a block vector of length $k$, 
and
$\diag(B_1,B_2,\dots,B_k)$ is a $k\times k$ block
diagonal matrix, in both cases with blocks $B_1,B_2,\dots,B_k$.

\item
 $W_{k,l}$ denotes the
$k\times l$ leading (that is, northwestern) block
of an $m\times n$ matrix $W$ for $k\le m$ and $l\le n$.
 A matrix is {\em strongly nonsingular} if
all its square leading blocks are nonsingular.

\medskip
\item
$\mathcal R(W)$, $W^T$ and $W^H$ denote its range (that is, column span),
 transpose 
and Hermitian 
 transpose, respectively. 
$W^H=W^T$ for a real matrix $W$.
\item
An $m\times n$ matrix $W$ is called
{\em unitary} (in the real case also and preferably {\em orthogonal})
if $W^HW=I_n$ or  $WW^H=I_m$. 
\item
 $Q(W)$ denotes the matrix 
obtained by means of column orthogonalization
of a matrix $W$, 
followed by the deletion of 
the columns filled with zeros (cf. \cite[Theorem 5.2.3]{GL13}).
\item
$||W||$ and
$||W||_F$ 
denote its spectral and Frobenius norms, respectively.
 
\item
 $S_W\Sigma_WT^T_W$ is its full Singular Value Decomposition (or full SVD) of an 
 $m\times n$matrix $W$  where 
 $S_W$ and $T_W$ are  square unitary matrices and 
 $\Sigma_W=\diag(\diag(\sigma_j(W))_{j=1}^{\rho},O_{m-\rho,n-\rho})$, 
for $\rho=\rank (W)$, is the $m\times n$ 
diagonal matrix of the singular values of the matrix $W$, 
 $$\sigma_1(W))\ge\sigma_2(W))\ge \cdots\ge \sigma_\rho(W))>0,~
\sigma_j(W))=0~{\rm for}~j>\rho.$$
\item
$W=S_{W,\rho}\Sigma_{W,\rho}T^T_{W,\rho}$ is its compact SVD,
where $\Sigma_{W,\rho}=\diag(\sigma_j(W))_{j=1}^{\rho}$ is the
$\rho\times \rho$ leading submatrix of $\Sigma_W$
and the matrices $S_{W,\rho}$ and $T_{W,\rho}$ are formed by the first $\rho$
columns of the matrices $S_W$ and $T_W$, respectively.
\item
 $W^+=T_{W,\rho}\Sigma_{W,\rho}^{-1}S^T_{W,\rho}$ 
is its Moore--Penrose pseudo
inverse. 

$(W^+)^+=W$, $WW^+=I_m$ if $\rank (W)=m$,
$W^+W=I_n$ if $\rank (W)=n$,
and $W^+=W^{-1}$ for a nonsingular matrix $W$.
\item
$||W||=\sigma_1(W)$ and
$||W||_F=(\sum_{j=1}^{\rho}\sigma_j^2(W))^{1/2}=$ 
(Trace $(W^HW))^{1/2}$ 
denote its spectral and Frobenius norms, respectively.

$||VW||\le ||V||~||W||$ and $||VW||_F\le ||V||_F||W||_F$,
for any matrix product $VW$.

$||U||=||U^+||=1$,
$||UW||=||W||$ and $||WU||=||W||$ if the matrix $U$ is unitary.
\item
$\sigma_{\rho}(W)=1/||W^+||$.
 $\kappa(W)=||W||~||W^+||=\sigma_1(W)/\sigma_{\rho}(W)\ge 1$
is the condition number
of  $W$. 
\item
The  $\xi$-rank of a matrix, for a
 positive  $\xi$,
is the minimum rank of its approximations within 
the norm bound  $\xi$. 
The   {\em numerical rank} of a matrix is its  $\xi$-rank
for $\xi$ 
small in context.
\item
A matrix $W$ is 
{\em ill-conditioned} if its condition number 
 is large in context
or equivalently if its rank exceeds its numerical rank.
The matrix is  {\em well-conditioned}
if  its condition number is reasonably bounded.
The ratio of the output and input error norms of Gaussian elimination 
is roughly the condition number of an input matrix
(cf.  \cite{GL13}). 
\end{enumerate}


\subsection{Random matrices: definitions, some basic properties}\label{srndmdef}


We use the acronym `` {\em i.i.d.}" for
``independent identically distributed", keep
referring to standard Gaussian random variables
  just as {\em Gaussian},
and call random variables {\em uniform} over a fixed finite set
if their values are sampled from this set
 under the uniform
probability distribution on it.

 The matrix is {\em  Gaussian}
if all its entries are i.i.d. Gaussian variables.


\begin{theorem}\label{thrnd} 
Suppose that $A$ is  a nonsingular  
$n\times n$ matrix  and $H$ is an $n\times n$ matrix 
whose  entries are linear combinations
of finitely many i.i.d. random variables
and let $\det((AH)_{l,l})$
vanish identically in them  for neither of the integers $l$, $l\le n$. 
(i) If the  variables are uniform
 over a set $\mathcal S$ of cardinality $|\mathcal S|$, then 
 the matrix $(AH)_{l,l}$ 
 is singular with a probability at most $l/|\mathcal S|$,  for any $l$, and so
 the matrix $AH$ is strongly nonsingular 
with a probability at least  
$1-0.5(n-1)n/|\mathcal S|$.  
(ii) If these i.i.d. variables
are  Gaussian, then 
the matrix $AH$ is strongly nonsingular 
with probability 1.  
\end{theorem}


\begin{proof}
claim (i) of the theorem follows from a celebrated lemma of \cite{DL78}. Derivation 
is specified,
e.g., in \cite{PW08}. Claim (ii)  follows because the equation $\det((AH)_{l,l})$
for any integer $l$ in the range from 1 to $n$ defines an algebraic variety of a lower
dimension in the linear space of the input variables
(cf. \cite[Proposition 1]{BV88}).
\end{proof}


\begin{lemma}\label{lepr3} ({\rm Orthogonal  invariance of a Gaussian matrix.})
Suppose that $k$, $m$, and $n$  are three  positive integers, $k\le \max\{m,n\}$,
$G$ is an 
 $m\times n$  Gaussian matrix, and
$S$ and $T$ are $k\times m$ and 
$n\times k$ orthogonal (or unitary) 
 matrices, respectively.
Then $SG$ and $GT$ are Gaussian matrices.
\end{lemma}


\subsection{Norm of a Gaussian matrix and of its pseudo inverse}\label{srndm}


Next we recall some estimates for the norms of a Gaussian matrix and of its
pseudo inverse. For simplicity we assume that we deal with real matrices,
but similar estimates in the case of complex 
 matrices can be found in 
  \cite{CD05} 
and \cite{ES05}.
Hereafter we write 
$\nu_{m,n}=||G||$ and
$\nu_{m,n}^+=||G^+||$
for  a  Gaussian $m\times n$ matrix  $G$,
and write $\mathbb E(v)$ for the expected value of 
a random variable $v$.


\begin{theorem}\label{thsignorm}
(Cf. \cite[Theorem II.7]{DS01}.)
Suppose 
that $m$ and $n$ are positive integers,
$h=\max\{m,n\}$, $t\ge 0$,  and
$z\ge 2\sqrt {h}$.
Then 
(i) {\rm Probability}$\{\nu_{m,n}>t+\sqrt m+\sqrt n\}\le
\exp(-t^2/2)$ and 

(ii) $\mathbb E(\nu_{m,n})\le \sqrt m+\sqrt n$.
\end{theorem}


\begin{theorem}\label{thsiguna} 
Let $\Gamma(x)=
\int_0^{\infty}\exp(-t)t^{x-1}dt$
denote the Gamma function and let  $x>0$. 
Then 

(i)  {\rm Probability} $\{\nu_{m,n}^+\ge m/x^2\}<\frac{x^{m-n+1}}{\Gamma(m-n+2)}$
for $m\ge n\ge 2$,

(ii) {\rm Probability} $\{\nu_{n,n}^+\ge x\}\le 2.35 {\sqrt n}/x$ 
for $n\ge 2$,

(iii) $\mathbb E((\nu_{F,m,n}^+)^2)=m/|m-n-1|$, provided that $m>n+1>2$,
 and

(iv) $\mathbb E(\nu^+_{m,n})\le e\sqrt{m}/|m-n|$, for $e=2.7182818\dots$,
provided that $m\neq n$.
\end{theorem}


\begin{proof}
 See \cite[Proof of Lemma 4.1]{CD05} for claim (i), 
\cite[Theorem 3.3]{SST06} for claim (ii),  and
\cite[Proposition 10.2]{HMT11} for claims (iii) and (iv).
\end{proof}


Probabilistic upper bounds 
on $\nu^+_{m,n}$ of  Theorem \ref{thsiguna} are reasonable already  
for square matrices, for which $m=n$,
but become much stronger as the difference $|m-n|$ grows above 2.


Theorems \ref{thsignorm} and \ref{thsiguna}
combined imply that an $m\times n$ Gaussian
matrix is very well-conditioned  
if the integer $m-n$ is large or even moderately large,
and still can be considered well-conditioned  
if the integer $|m-n|$ is small or even vanishes
(possibly with some grain of salt in the latter case).
These properties are immediately extended to all submatrices because
they are also  Gaussian.


\section{Recursive Factorization of a Matrix.  BGE and GENP
}\label{sbgegenpn}


In this section we specify recursive factorization of a strongly nonsingular matrix, which we will use as a  basis
for our simultaneous study of safety and numerical safety of GENP and BGE.

For a 
nonsingular
$2\times 2$ block matrix $A=\begin{pmatrix}
B  &  C  \\
D  &  E
\end{pmatrix}$ of size $n\times n$
with nonsingular $k\times k$ {\em pivot block} $B=A_{k,k}$, 
 define 
$S=S(A_{k,k},A)=E-DB^{-1}C$,
the {\em Schur complement} of $A_{k,k}$ in $A$,
and
the block factorizations,

\begin{equation}\label{eqgenp}
\begin{aligned}
A=\begin{pmatrix}
I_k  &  O_{k,r}  \\
DB^{-1}  & I_r
\end{pmatrix}
\begin{pmatrix}
B  &  O_{k,r} \\
O_{r,k}  &  S
\end{pmatrix}
\begin{pmatrix}
I_k  &  B^{-1}C  \\
O_{k,r}  & I_r
\end{pmatrix}
\end{aligned},
\end{equation} 
\begin{equation}\label{eqgenpin}
\begin{aligned}
A^{-1}=\begin{pmatrix}
I_k  &  -B^{-1}C  \\
O_{k,r}  & I_r
\end{pmatrix}
\begin{pmatrix}
B^{-1}  &  O_{k,r} \\
O_{r,k}  &  S^{-1}
\end{pmatrix}
\begin{pmatrix}
I_k  &  O_{k,r}  \\
-DB^{-1}  & I_r
\end{pmatrix}
\end{aligned}.
\end{equation} 

These factorizations represent 
Gauss-Jordan elimination applied to a
$2\times 2$ block matrix.

We readily verify that $S^{-1}$ is the $(n-k)\times (n-k)$ trailing
(that is, southeastern) block of the inverse matrix $A^{-1}$, and so  
the Schur complement $S$ is nonsingular since the matrix $A$ is nonsingular.

Factorization (\ref{eqgenpin}) reduces the inversion of the matrix $A$ 
to the inversion of the leading block $B$ and its 
Schur complement $S$, and we can 
recursively reduce the inversion task to the case of the leading blocks 
and Schur complements of decreasing sizes 
as long as the leading blocks are nonsingular.  
After sufficiently many  recursive steps 
of this
process of
BGE,
we only need to invert matrices  of
small sizes, and then we can stop the process   
 and  apply a selected black box
inversion algorithm, e.g., based on orthogonalization. 
If we limit the number of recursive steps, we arrive at BGE 
dealing with large blocks and can use  the 
benefits  of block matrix algorithms, but if we keep
recursive partitioning, then BGE eventually turn into GENP.  

Namely, in $\lceil\log_2(n)\rceil$ recursive steps
all pivot blocks and 
all other
matrices involved into the 
resulting factorization
turn into scalars, 
all matrix 
multiplications and inversions turn into 
scalar multiplications and divisions,
and we arrive at a
{\em complete recursive factorization} of the matrix $A$.
If
$k=1$ at all recursive steps, then the complete
recursive factorization (\ref{eqgenpin})
defines GENP.

Moreover, any complete recursive factorizations
turns into GENP up to the order in which we 
consider its steps.
This follows because  at most $n-1$  distinct
Schur complements $S=S(A_{k,k},A)$, 
for $k=1,\dots,n-1$,
are involved in all recursive block factorization
processes for $n\times n$ matrices $A$,
and so we arrive at the same Schur complement in a fixed position
via GENP and via any other recursive block factorization (\ref{eqgenp}).
Hence we can interpret factorization step  (\ref{eqgenp})
as the block elimination of the
first $k$
columns of the matrix $A$, 
which produces  the  matrix $S=S(A_{k,k},A)$.
If the dimensions 
$d_1,\dots,d_r$ and $\bar d_1,\dots,\bar d_{\bar r}$ of 
the pivot  blocks in 
two block elimination processes
 sum to the same integer $k$, that is, if 
$k=d_1+\dots+d_r=\bar d_1+\dots+\bar d_{\bar r}$,
then 
both processes produce the same Schur complement $S=S(A_{k,k},A)$.
The following results extend this observation. 


\begin{theorem}\label{thsch}
In the recursive block factorization process based on (\ref{eqgenp}),  
 the diagonal block and its Schur complement 
in every block diagonal factor is either 
a leading block of the input matrix $A$ or the Schur complement $S(A_{h,h},A_{k,k})$
 for some integers $h$ and $k$ such that $0<h<k\le n$ and
$S(A_{h,h},A_{k,k})=(S(A_{h,h},A))_{h,h}$.
\end{theorem}


\begin{corollary}\label{corec}
The complete recursive block factorization process based on equation (\ref{eqgenp})  
can be computed by involving no singular
pivot blocks (and, in particular, no pivot elements vanish)
if and only if the input matrix $A$ is strongly nonsingular.
\end{corollary}
\begin{proof}
Combine Theorem \ref{thsch} with the equation 
$\det A=(\det B)\det S$,
implied by (\ref{eqgenp}).
\end{proof}


\section{Multiplicative Pre-processing for GENP and BGE}\label{s0}


\subsection{Definition and criteria of safety and numerical safety}\label{svrf}

 In this section, $A$ denotes a nonsingular $n\times n$ matrix.

Suppose that 
 the vector ${\bf y}=A{\bf b}$ satisfies
  pre-processed 
linear systems $AH{\bf y}={\bf b}$ and  $FAH{\bf y}=F{\bf b}$.
Then
the vector 
 ${\bf x}=H{\bf y}$ for ${\bf y}=A^{-1}{\bf b}$ satisfies
both linear systems $A{\bf x}={\bf b}$
and
$FA{\bf x}=F{\bf b}$.
We are going to study such
pre-processing $A\rightarrow AH$ for  GENP and BGE with random and fixed post-mul\-ti\-pli\-ers $H$.
Our analysis is immediately extended to the
pre-processing maps $A\rightarrow FA$, $A\rightarrow FAH$, and 
 $A\rightarrow FAF^T$.



 We call GENP and BGE {\em safe} if 
they proceed to the end with no divisions by 0. 

Corollary \ref{corec} implies the following result for computations in
 any field (cf. Remark \ref{rebge_genp}).


\begin{theorem}\label{th0} 
GENP is safe 
if and only if the input matrix is strongly nonsingular. 
\end{theorem}


Next assume that GENP and BGE are performed  numerically,
with rounding to a fixed precision, e.g.,
the IEEE standard double precision. 
Then extend the concept of safe GENP  and BGE to
{\em numerically safe GENP and BGE} by requiring that
the input matrix be strongly nonsingular and {\em strongly well-conditioned},
that is, that the matrix itself and all its square  leading blocks 
 be nonsingular and well-conditioned.
Any inversion algorithm for a nonsingular matrix 
is highly sensitive to the
input and rounding errors if the matrix is 
ill-conditioned \cite{GL13}.
 GENP explicitly or implicitly involves the inverses 
of all its square leading blocks, and we arrive at the following 
{\em Criterion
of Numerical Safety of GENP}
implied by \cite[Theorem 5.1]{PQZ13}:

{\em  GENP applied to a strongly nonsingular matrix is highly sensitive to the
input and rounding errors if and only
 if some of the square leading blocks are ill-conditioned}.

Let us restate this criterion in the form similar to Theorem \ref{th0}. 


\begin{theorem}\label{th011} 
GENP is safe and numerically safe 
if and only if the input matrix is strongly nonsingular
and  strongly well-conditioned. 
\end{theorem}


\begin{remark}\label{rebge_genp}
BGE is safe if so does GENP.
Likewise BGE is safe numerically  if so does GENP. 
Thus our 
proofs of safety and numerical safety of GENP
apply to BGE. The converse is not true: 
GENP fails (resp. fails numerically) if any  square
leading block of the input matrix is singular (resp. ill-conditioned),
but BGE may by-pass this block and be safe (resp. numerically safe).
\end{remark}


 

\subsection{GENP with
Gaussian pre-processing
}\label{sgspr}  
 

Next we prove that GENP with
Gaussian pre-processing is safe with probability 1 for any nonsingular input  matrix   and is likely to be  numerically  safe if this matrix is also well-conditioned.                                                 

\begin{theorem}\label{thdgr} 
Assume that we are given a nonsingular and well-conditioned $n\times n$ matrix $A$ 
and a pair of $n\times n$  
 Gaussian matrices $F$ and $H$
and let  
$\nu_{k,k}^+$, $\nu_{k,n}^+$, and $\nu_{n,k}^+$ denote
 random variables
of Section \ref{srndm}.
Then

(i) the matrices $FA$, $AH$, and $FAH$ are strongly nonsingular 
with a probability 1,


(ii) $||((AH)_{k,k})^+||\le \nu_{k,k}^+||A_{k,n}^+||\le \nu_{k,k}^+||A^+||$,
$||((FA)_{k,k})^+|| \le \nu_{k,k}^+||A_{n,k}^+||\le \nu_{k,k}^+||A^+||$,

(iii)  $||((FAH)_{k,k})^+||\le
\nu_{k,k}^+\nu_{n,k}^+||A^+||$,
and 

(iv) this bound is tight where $A$ is the identity matrix, that is, 
$||((FH)_{k,k})^+||=\nu_{k,k}^+\nu_{n,k}^+$.
 
\end{theorem} 


\begin{proof}

Claim (i) follows from claim (ii) of Theorem \ref{thrnd}.

Hereafter a pair of subscripts  $p,q$ shows the matrix size $p\times q$.
The proof of claim (ii) is similar for both products $AH$ and $FA$;
we only cover the case of the former one.
  
Notice that $(AH)_{k,k}=A_{k,n} H_{n,k}$ and
 substitute  compact
SVD $A_{k,n}=S_{k,k}\Sigma_{k,k} T_{n,k}^T$
where $\Sigma_{k,k}$ is a diagonal matrix and
$S_{k,k}$ and $T_{n,k}$ are orthogonal matrices.
Obtain 
$$(AH)_{k,k}=S_{k,k}\Sigma_{k,k} T_{n,k}^TH_{n,k}=S_{k,k}\Sigma_{k,k} G_{k,k}$$
where
$G_{k,k}=T_{n,k}^TH_{n,k}$ is a $k\times k$
 Gaussian matrix by virtue of 
Lemma \ref{lepr3}.
Deduce that
$$((AH)_{k,k})^+=G_{k,k}^+\Sigma_{k,k}^{-1}S_{k,k}^T,~{\rm and ~so}~
||((AH)_{k,k})^+||=||G_{k,k}^+\Sigma_{k,k}^{-1}||\le ||G_{k,k}^+||~||\Sigma_{k,k}^{-1}||.$$
Substitute the equations 
  $||G_{k,k}^{+}||=\nu_{k,k}^+$ and 
$||\Sigma_{k,k}^{-1}||=||A_{k,n}^+||\le ||A^+||$
and obtain claim (ii).

Let us prove claim (iii). Notice that $(FAH)_{k,k}=F_{k,n} A H_{n,k}$, 
 substitute full
SVD \\
$A=S\Sigma T^T,$ and
obtain  
$$(FAH)_{k,k}=F_{k,n} S\Sigma T^T H_{n,k}
=G'\Sigma G''$$
where $G'=G_{k,n}=F_{k,n} S$ and $G''=G_{n,k}=T^T H_{n,k}$
are Gaussian matrices
by virtue of  Lemma \ref{lepr3}.

Assume that the matrix 
$\Sigma G''$ has full rank
(it has it with probability 1 by virtue of 
Theorem \ref{thrnd}) and deduce that
$(G''^+\Sigma^+)$ is the left inverse of  the matrix $\Sigma G''$ and hence
$(G''^+\Sigma^+)=(\Sigma G'')^+$.

Therefore 
$||(\Sigma G'')^+||\le 
||\Sigma^+||~||G''^+||=||A^+||~\nu_{n,k}$.
Apply claim (ii) for $A=\Sigma G''$ and $F=G'$ and obtain claim (iii).

Let us prove claim (iv).
 
Notice that
$||((FH)_{k,k})^+||=
||(F_{n,k}H_{k,n})^+||$
and let $H_{k,n}=S_H\Sigma_HT_H^T$ be SVD
where $S_H$ is an $n\times k$ orthogonal matrix, 
 $\Sigma_H$ and $T_H$ are 
$k\times k$
nonsingular matrices,  
$||\Sigma_H^{-1}||=\nu_{n,k}$ and  
$||T_H^{-1}||=1$.

Observe that 
$F_{n,k}S_H$ is a  $k\times k$
Gaussian matrix by virtue of  Lemma \ref{lepr3}, and so 
$||(F_{n,k}S_H)^+||=\nu_{k,k}$.
Finally observe that 
$$||((FH)_{k,k})^+||=
||(F_{n,k}S_H\Sigma_HT_H^T)^+||=
||(FS_H)^+||~||\Sigma_H^{-1}||~||T_H^{-1}||$$
and substitute 
the equations 
$||(F_{n,k}S_H)^+||=\nu_{k,k}$,
$||\Sigma_H^{-1}||=\nu_{n,k}$, and  
$||T_H^{-1}||=1$.
\end{proof}


\begin{remark}\label{redgr}
\cite[Corollary 5.2]{PQY15} provides a correct, although very long proof of claim (ii)
of Theorem \ref{thdgr}
in the case of post-multiplication by $H$, but states the result
with an error, by writing $\nu_{n,k}$ instead of the correct $\nu_{k,k}$.
We fix the statement and include a short proof for the sake of completeness of our presentation.
Claims (iii) and (iv) of the theorem are new.
\end{remark}


 Theorems  \ref{thsignorm},  \ref{thsiguna}, and \ref{th0}--\ref{thdgr}
 together
imply  
the following result.


\begin{corollary}\label{cogpr}   
GENP  is safe with probability 1 and is likely to be
numerically  safe if it is applied to the matrices $FA$, $AH$ and $FAH$
where $A$ is a nonsingular and well-conditioned $n\times n$  matrix
and $F$ and $H$ are Gaussian $n\times n$  matrices.
\end{corollary} 

Next we observe some additional benefits of  BGE.

\begin{remark}\label{reblock}
BGE 
can  use
additional benefits of block matrix algorithms and, 
rather surprisingly, 
 of saving random  variables and flops. 
E.g., first pre-process  
 the $k\times k$ leading block of the input matrix
for a proper integer $k<n$ by using 
$n\times k$  Gaussian multipliers. Having factored this block,
 decrease the input size from $n$ to $n-k$ and then 
 re-apply 
Gaussian pre-processing. Already by using such a two-step 
block pre-processing for $k=n/2$, we 
save 1/4 of all random  variables and 3/8 of arithmetic operations involved.
\end{remark}





\subsection{GENP with
any nonsingular and well-conditioned pre-processing is safe and numerically  safe
on the average input}\label{sgsprav}  
 

In  the previous subsection we assumed that an input matrix  $A$ is fixed, and the multipliers $F$ and $H$ are Gaussian.
Next we prove {\em a dual theorem} where  we assume that the multipliers are fixed and the input matrix is Gaussian. We 
obtain  probabilistic estimates for 
the matrices $(AH)_{k,k}$,
 $(FA)_{k,k}$, and $(FAH)_{k,k}$
 similar to those of Theorem \ref{thdgr}, and we will readily extend them to   the estimates for 
pre-processed GENP applied to the average input matrix (see Corollary \ref{cogprd}).

\begin{theorem}\label{thdgrd} 
Assume that we are given a Gaussian
$n\times n$ matrix $A$ 
and a pair of $n\times n$ 
nonsingular and well-conditioned 
 matrices $F$ and $H$
and let  
$\nu_{k,k}^+$ denote a
random variable 
of Section \ref{srndm}.
Then

(i) the matrices $FA$, $AH$, and $FAH$ are strongly nonsingular 
with a probability 1,


(ii)  $||((AH)_{k,k})^+||\le \nu_{k,k}^+||H_{n,k}^+||$,
$||((FA)_{k,k})^+||\}
\le \nu_{k,k}^+||F_{k,n}^+||$, and

(iii) $||((FAH)_{k,k})^+||\le ||F_{k,n}^+|| \nu_{k,k}^+ ||H_{n,k}^+||\le
||F^+||\nu_{k,k}^+||H^+||$.
 
\end{theorem} 


\begin{proof}
The proof is similar to that of Theorem \ref{thdgr}, and
we only specify the proof of claim (iii). 
Recall that $(FAH)_{k,k}=F_{k,n} A H_{n,k},$
 substitute compact
SVDs $F_{k,n}=S_{F,k,k}\Sigma_{F,k,k} T_{F,n,k}^T$  and 
$H_{n,k}=S_{H,n,k}\Sigma_{H,k,k} T_{H,k,k}^T,$
and obtain 
$(FAH)_{k,k}=S_{F,k,k}\Sigma_{F,k,k}G_{k,k}\Sigma_{H,k,k} T_{H,k,k}^T$.
Here $G_{k,k}=T_{F,k,n}^TAS_{H,n,k}$ is a $k\times k$ Gaussian matrix 
by virtue of  Lemma \ref{lepr3}.
The matrices $S_{F,k,k}$, $\Sigma_{F,k,k}$, $\Sigma_{H,k,k}$,
$T_{H,k,k}$ are nonsingular by assumption,
and so is the matrix $G_{k,k}$ with probability 1
by virtue of claim (ii) of Theorem \ref{thrnd}.
Hence with  probability 1, 
$((FAH)_{k,k})^+=T_{H,k,k}\Sigma_{H,k,k}^{-1} G_{k,k}^{-1} \Sigma_{F,k,k}^{-1} S_{F,k,k}^T,$ and so
$||((FAH)_{k,k})^+||\le ||T_{H,k,k}||~||\Sigma_{H,k,k}^{-1}||~||G_{k,k}^{-1}||
~||\Sigma_{F,k,k}^{-1}||~||S_{F,k,k}^T||.$
Substitute 

$||T_{H,k,k}||=||S_{F,k,k}^T||=1$,
$||\Sigma_{F,k,k}^{-1}||=||(F_{k,n}^+||$,
$||\Sigma_{H,k,k}^{-1}||=||H_{n,k}^+||$,
and  $||G_{k,k}^{-1}||=\nu_{k,k}^+$.
\end{proof}

Theorems \ref{th011}, \ref{thdgr}, \ref{thsignorm}, and \ref{thsiguna}
 together
imply  
the following result.


\begin{corollary}\label{cogprd}   
GENP is safe and
numerically  safe when it is applied to
the matrices $FA$, $AH$ and $FAH$
where $A$ is the
 average $n\times n$ matrix 
defined under the Gaussian probability distribution 
and $F$ and $H$ are $n\times n$ nonsingular and well-conditioned  matrices.
\end{corollary} 
 

\subsection{Heuristic amelioration of pre-processing for GENP}\label{sham} 


  Corollary \ref{cogprd} implies 
that application of GENP with  
 pre-processing by a nonsingular and well-conditioned multipliers $F$ and/or $H$  is safe and numerically safe 
 for most of  nonsingular and well-conditioned input matrices $A$.
 This somewhat informal claim is 
 in good accordance with 
our empirical study, although
for any  multipliers $F$ 
we can readily exhibit bad nonsingular and well-conditioned inputs $A$
for which GENP applied to the matrix $FA$ fails numerically,
and similarly for GENP applied to  matrices $AH$ and $FAH$ for any fixed multiplier $H$.

Next we comment on two heuristic recipes for choosing multipliers.

(i) Clearly the product of two sparse matrices
has good chances to have singular square leading blocks,
and so one can be cautious about pre-processing 
a sparse matrix with
sparse multipliers. For a partial remedy,
we can  more evenly distribute  nonzero entries
 throughout a sparse multiplier, but 
for a more reliable remedy, we can apply dense structured multipliers.

(ii) Here is  a general useful heuristic recipe for simplifying repeated pre-processing when
GENP has failed numerically for two
matrices $AH_1$ and $AH_2$:  apply GENP 
to the sum $AH_1+AH_2$  or product $AH_1H_2$.
In our tests in Section \ref{sexp1}, this recipe has consistently worked, but there are also other attractive options, e.g.,  using  linear combinations or polynomials in $H_1$ and $H_2$ as multipliers. 


\section{Two Classes of Basic Structured Matrices for Generation of Efficient Multipliers}\label{sbstr}  


\subsection{Matrices of discrete Fourier transform}\label{stpldft} 


\begin{definition}\label{defdft}

 Write 
$\omega=\exp(\frac{2\pi\sqrt{-1}}{n})$,
$\Omega=\Omega_n=(\omega^{ij})_{i,j=0}^{n-1}$,
 $\frac{1}{\sqrt n}\Omega$ is unitary,
$\Omega^{-1}=\frac{1}{n}\Omega^H$,
$\omega$ denotes a primitive $n$-th root of unity,
$\Omega$ and $\Omega^{-1}$
 denote the matrices of the discrete Fourier transform at $n$ points
and its inverse, to which we refer as
DFT$(n)$ and IDFT$(n)$, respectively. 

\end{definition}


\begin{remark}
If $n=2^k$ is a power of 2, we can apply the FFT algorithm and perform
DFT$(n)$ and IDFT$(n)$  
by using only $1.5n\log_2(n)$ and $1.5n\log_2(n)+n$
arithmetic operations, respectively. 
For an $n\times n$ input and any $n$,
we can perform DFT$(n)$ and  IDFT$(n)$
by using $cn\log(n)$ arithmetic operations,
but for a larger constant $c$
(see \cite[Section 2.3]{P01}).
\end{remark}



\subsection{Circulant  and  $f$-circulant matrices}\label{stplcrc} 




For a positive integer $n$ and a complex scalar $f$, define 
 the $n\times n$ unit $f$-circulant
matrix $Z_f=\begin{pmatrix}  
       0  & ~f\\
        
        I_{n-1}   & ~0 
    \end{pmatrix}$  and  the $n\times n$ general 
$f$-{\em circulant matrix} $Z_f({\bf v})=\sum_{i=0}^nv_iZ_f^i$, 

$$Z_f=\begin{pmatrix}
        0  ~ & \dots&     ~  ~\dots   ~ ~& 0 & ~f\\
        1   & \ddots    &   &    ~ 0 & ~0\\
        \vdots         & \ddots    &   \ddots & ~\vdots &~ \vdots  \\
         \vdots   &    &   \ddots    &~ 0 & ~0 \\
        0   & ~  \dots &      ~ ~~  \dots   ~~ ~ ~ & 1 & ~0 
    \end{pmatrix}~~{\rm and}~~
Z_f({\bf v})=\begin{pmatrix}v_0&fv_{n-1}&\cdots&fv_1\\ v_1&v_0&\ddots&\vdots\\ \vdots&\ddots&\ddots&fv_{n-1}\\ v_{n-1}&\cdots&v_1&v_0\end{pmatrix}.$$

$Z=Z_0$ is the unit {\em down-shift matrix}. $Z_0({\bf v})$ is a lower triangular Toeplitz matrix, $Z_1({\bf v})$
is a {\em  circulant} matrix.
$Z_f^n=fI_n$, 
and the matrix $Z_f({\bf v})$ is defined by its first column ${\bf v}=(v_i)_{i=0}^{n-1}$.

We call an $f$-circulant matrix a {\em Gaussian $f$-circulant}
 (or just 
{\em Gaussian circulant} if $f=1$) 
 if its first column is filled  with 
independent Gaussian variables.
For every fixed $f$, the $f$-circulant matrices form an algebra 
in the linear space of $n\times n$   Toeplitz matrices 
\begin{equation}\label{eqtz}
T=(t_{i-j})_{i,j=0}^{n-1}. 
\end{equation}


Hereafter, for a vector 
${\bf u}=(u_i)_{i,j=0}^{n-1}$,  write
$D({\bf u})=\diag (u_0,\dots,u_{n-1})$,
that is,  $D({\bf u})$ is the diagonal matrix 
with the diagonal entries $u_0,\dots,u_{n-1}$.

\begin{theorem}\label{thcpw} (Cf. \cite[Theorem 2.6.4]{P01}.)
If $f\neq 0$, then $f^n$-circulant matrix $Z_{f^n}({\bf v})$ 
of size $n\times n$ 
can be factored as follows,
$ Z_{f^n}({\bf v})=U_f^{-1}D(U_f{\bf v})U_f~{\rm for}~
U_f=\Omega D({\bf f}),~~{\bf f}=(f^i)_{i=0}^{n-1},~{\rm and}~f\ne 0$.
In particular, 
for circulant matrices, 
 $D({\bf f})=I$, $U_f=\Omega$,
and  
$Z_1({\bf v})=\Omega^{-1}D(\Omega{\bf v})\Omega$.
\end{theorem}


\begin{remark}\label{rettm}
We cannot extend this theorem  directly to a $k\times k$ $0$-circulant
(that is, triangular Toeplitz) matrix $T$,
but 
we can embed this matrix (as well as 
 any $k\times k$
Toeplitz matrix $T$) into an $n\times n$
 circulant matrix $C$ for $n\ge 2k-1$ and then can readily extract the product 
 $T{\bf v}$
 (for any vector ${\bf v}$) from the product
$C{\bf w}$ for a proper vector ${\bf w}$
having  ${\bf v}$ as a subvector.
\end{remark}

\begin{remark}\label{rettmc}
For an $n\times n$ Toeplitz or Toeplitz-like 
matrix $A$ and an $n\times n$ circulant matrix $H$, one can compute 
a standard displacement representation of the product $AH$ 
by applying just $O(n\log(n))$ flops (see definitions and derivation in \cite{P01}).
\end{remark}


\begin{corollary}\label{cocrc}

(i) If ${\bf u}=\Omega {\bf v}$,
then ${\bf v}=\frac{1}{n}\Omega^H {\bf u}$.
 
 (ii) If the vector ${\bf v}$
is  Gaussian, then 
 so is 
also
the vector ${\bf u}=(u_i)_{i=1}^n=\frac{1}{\sqrt n}\Omega {\bf v}$
(by virtue of Lemma \ref{lepr3}) and vice versa.
Each of the two vectors defines a Gaussian circulant  
matrix  $Z_1({\bf v})$.

(iii) By choosing 
$u_i=\exp(\frac{\phi_i}{2\pi}\sqrt {-1})$ and real 
Gaussian variable $\phi_i$ for all $i$, we arrive at
 a random  real orthogonal or unitary  $n\times n$ circulant matrix $Z_1({\bf v})$
defined by $n$ real Gaussian parameters $\phi_i$,
$i=0,\dots,n-1$.
Alternatively we   can set $\phi_i=\pm 1$
for all $i$ and choose the signs $\pm$
at  random, with i.i.d. probability 1/2 for all signs. 

(iv) By adding another Gaussian parameter $\phi$, 
we can define
 a  random  real orthogonal or unitary $f$-circulant matrix $Z_f({\bf v})$
for $f=\exp(\frac{\phi}{2\pi}\sqrt {-1})$. 
\end{corollary}

The following results 
imply that a Gaussian circulant matrix is likely to be
well-conditioned.

\begin{theorem}\label{thcrcnc} (Cf. \cite{PSZ15}.)
Suppose that $Z_1({\bf v})=\Omega^H D\Omega$ 
is a nonsingular circulant $n\times n$ matrix, and
let $D({\bf g})=\diag(g_i)_{i=1}^n$, for ${\bf g}=(g_i)_{i=1}^n$.
Then $||Z_1({\bf v})||=\max_{i=1}^n |g_i|$,
$||Z_1({\bf v})^{-1}||=\min_{j=1}^n |g_j|$,
and $\kappa(Z_1({\bf v}))=\max_{i,j=1}^n |g_i/g_j|$,
for ${\bf v}=\Omega^{-1}{\bf g}$.
\end{theorem}

\begin{remark}\label{reblck1}
Suppose that a circulant matrix $Z_1({\bf v})$
has been defined by its first column vector ${\bf v}$ 
filled with  integers $\pm 1$ for a 
random choice of the i.i.d. signs $\pm$, each $+$
and $-$ chosen with probability $1/2$.
Then, clearly, 
the entries $g_i$ of the vector 
${\bf g}=\Omega{\bf v}=(g_i)_{i=1}^n$
satisfy $|g_i|\le n$ for all $i$ an $n$,
and furthermore, 
with a probability close to 1,
$\max_{i=1}^n\log(1/|g_i|)=O(\log (n))$
as $n\rightarrow \infty$. 
\end{remark}


\begin{remark}\label{reblck}
 In the case of a Gaussian circulant matrix $Z_1({\bf v})$, 
all the entries $g_i$  are i.i.d.  Gaussian variables,
and the condition number $\kappa(Z_1({\bf v}))=\max_{i,j=1}^n |g_i/g_j|$
is not likely to be large.
\end{remark}

\begin{remark}\label{renmrmba} 
The superfast but numerically unstable
 MBA  algorithm (cf. \cite{B85}, \cite[Chapter 5]{P01}) is  precisely the recursive BGE,
accelerated by means of exploiting the  
displacement structure of the input matrix
 throughout 
the recursive process of BGE.
Fortunately,  pre-processing with appropriate random structured multipliers 
is likely to fix these problems
keeping the algorithm superfast,
that is, using almost linear number of flops
(cf. \cite[Sections 5.6 and 5.7]{P01}). 
\end{remark}


\section{Generation of Efficient Multipliers}\label{ssprsml}  


\subsection{What kind of multipliers do we seek?}\label{srqm}


Trying to support safe and numerically safe GENP
we seek  multipliers with the following properties:

\begin{enumerate}
\item
Multipliers $F$ and $H$ must be  nonsingular and well-conditioned.
\item
The cost of the computation of the product $FA$, $AH$, $FAH$ or $FAF^H$
should be small. 
\item
Random multipliers  
should be generated  by using fewer random  variables. 
\item
For structured input matrices, the multipliers should have consistent structure.
\item
The multipliers should enable GENP to produce accurate   
 output with a probability close to 1 even with using 
no iterative refinement, at least for the inputs 
of reasonably small sizes.

\end{enumerate}

These rules give  general guidance but are a subject to 
 trade-off: 
e.g., by filling multipliers with values 0, 1, and $-1$, 
we save flops when we 
 compute their products with the input matrix $A$, 
but we increase the chances for success of  
our pre-processing 
if instead we fill the multipliers with real or complex random variables.

Some families of our new  nonsingular multipliers for GENP and BGE refine 
(towards the above properties)
the customary families of sparse and structured rectangular multipliers used for low-rank approximation (cf. \cite{HMT11}, \cite{M11}, \cite{W14}).  E.g., we sparsify the SRHT and  SRFT multipliers, which substantially
simplifies their generation and application as multipliers.

\begin{remark}\label{reref}
Property 5
 was satisfied 
in our tests
for most of our multipliers (unlike the multipliers of
\cite{BDHT13}  and \cite{DDF14}).
This property
  is important for $n\times n$
input matrices of smaller sizes: 
 refinement iteration involves $O(n^2)$ flops
versus cubic cost
of $\frac{2}{3}n^3$ flops, 
involved in Gaussian elimination,
and for small $n$ quadratic cost of refinement can make up  a 
large share of the overall cost. 
A single 
refinement iteration was always sufficient
(and more frequently was not even needed) in our tests
in order to match or to exceed the output accuracy of GEPP
(see also similar empirical data in \cite{PQZ13},
 \cite{BDHT13}, \cite{DDF14}, and \cite{PQY15}
and see \cite[Chapter 12]{H02}, \cite[Section 3.5.3]{GL13},
and the references therein for detailed coverage of 
iterative refinement).
\end{remark}


\subsection{Some basic  matrices for the generation of multipliers}\label{sdfcnd}


We are going to use the following matrix classes in the next subsections as our building blocks for defining 
efficient multipliers.


\begin{enumerate}
  \item
$P$ denotes an $n\times n$ 
 {\em permutation matrix}.
\item
$D$ denotes a unitary or real orthogonal {\em diagonal matrix}  $\diag(d_i)_{i=0}^{n-1}$, with 
fixed  or random
diagonal entries $d_i$ such that
$|d_i|=1$ for all $i$, and so each of $n$ entries $d_i$ lies 
on the unit circle $\{x:~|z|=1\}$, 
being either nonreal or  $\pm 1$.
\item
Define a {\em Givens rotation  matrix} $G(i,j,\theta)$ (cf. \cite{GL13}): 
for two integers $i$ and $j$,
$1\le i<j\le n$, and a 
fixed or random
real $\theta$, $0\le \theta \le 2\pi$, 
 replace
 the 
submatrix $I_2$ 
in the 
$i$th and $j$th rows and
columns   
of the
identity 
matrix $I_n$  
by the matrix  
$\big (\begin{smallmatrix} c  & s\\
-s   & c\end{smallmatrix}\big )$ 
where $c=\cos \theta$,
$s=\sin \theta$,  $c^2+s^2=1$.
\item
Define a {\em Householder reflection matrix}
$I_n-\frac{2{\bf h}{\bf h}^T}{{\bf h}^T{\bf h}}$
by a 
fixed or random 
vector ${\bf h}$ (cf. \cite{GL13}).
\end{enumerate}

We also use a DFT matrix $\Omega_n$
and transforms DFT, DFT$(n)$ and IDFT$(n)$ of Section \ref{stpldft}
as well as circulant and $f$-circulant matrices of Section  \ref{stplcrc}.
Besides unitary and orthogonal diagonal matrices $D$, defined above,
we use a nonsingular and well-conditioned  diagonal matrix 
$\diag(\pm 2^{b_i})$ in Family 3
of Section \ref{sexp1}, for random integers $b_i$
uniformly chosen from 0 to 3.


\subsection{Four families of dense structured multipliers
}\label{sinvchh}


{\bf FAMILY 1.} {\em The inverses of bidiagonal matrices:}
$$H=(I_n+DZ)^{-1}~{\rm or}~H=(I_n+Z^TD)^{-1}$$ for 
a diagonal matrix $D$   and the down-shift matrix
$Z$ of Section \ref{stplcrc}. 

We can randomize a matrix $H$ 
 by choosing $n-1$ random diagonal entries of
the matrix $D$
(whose leading entry  makes no impact on $H$) and 
 can pre-multiply it  by
a vector by using $2n-1$ flops.

$||H||\le \sqrt {n}$ because nonzero entries of the  
triangular 
matrix $H=(I_n+DZ)^{-1}$ have absolute values 1,  and
clearly $||H^{-1}||=||I_n+DZ||\le \sqrt 2$. Hence 
$\kappa(H)=||H||~||H^{-1}||$
 (the spectral condition number of  $H$) cannot exceed
$\sqrt {2n}$ for $H=(I_n+DZ)^{-1}$,
and likewise for $H=(I_n+Z^TD)^{-1}$.

\medskip

{\bf FAMILY 2.} {\em Chains of scaled  and permuted Givens rotations with a DFT factor}.
A permutation matrix $P$  and a sequence of angles $\theta_1,\dots,\theta_{n-1}$ 
together define a permuted chain 
of Givens rotations  
$$G(\theta_1,\dots,\theta_{n-1})=P\prod_{i=1}^{n-1}G(i,i+1,\theta_i)~
{\rm for}~n=2^k.$$
Combine two such chains $G_1$ and $G_2$
with scaling and DFT into the following dense  unitary matrix,
$$H=\frac{1}{\sqrt n}D_1G_1D_2G_2D_3\Omega_n$$ (cf. \cite[Section 4.6]{HMT11}),
for diagonal matrices
$D_1$, $D_2$ and $D_3$.
This  matrix can be multiplied by a vector by using about $10n$  flops
plus the cost of performing DFT.
 By randomizing diagonal and Givens  rotation factors  of $H$  
we can involve  up to $7n-2$ random  variables.

{\bf FAMILY 3.} {\em Pairs of scaled  and permuted Householder reflections  with a DFT factor.}
Define the unitary matrix 
$$H=\frac{1}{\sqrt n}D_1R_1D_2R_2D_3\Omega_n$$ 
where $D_1$,  $D_2$, and  $D_3$ are diagonal matrices,  $R_1$ and  $R_2$
are Householder reflections, and $\Omega_n$ is a DFT matrix of Section  \ref{stpldft}.
This matrix can be multiplied by a vector by using about $7n$  flops
plus the cost of performing DFT. By randomizing diagonal and 
 Householder reflection factors  of $H$
we can involve  up to $6n$ random  variables.

{\bf FAMILY 4.} $f$-{\em circulant matrices} of Section \ref{stplcrc}.
$H=Z_f({\bf v})=\sum_{i=0}^{n-1}v_iZ_f^i$, 
for  
the matrix $Z_f$ of $f$-circular shift,
defined by a 
scalar $f\neq 0$  and 
its first column vector ${\bf v}=(v_i)_{i=0}^{n-1}$.
By randomizing this vector we can involve up to $n$ random  variables,
and then 
Theorem \ref{thcrcnc} and 
 \cite{PSZ15} enable us to
 bound the condition number $\kappa(H)$. 
By virtue of Theorem \ref{thcpw},  $Z_f({\bf v})=(FD_f)^{-1}DFD_f$ where
$D$ is a diagonal matrix,  $D_f=\diag(f^i)_{i=0}^{n-1}$, and $\Omega_n$ is 
the DFT matrix  
of  Section \ref{stpldft}.
Based on this expression, we can multiply the  matrix $H$ 
 by applying two DFT($n$), an IDFT($n$),
and  additionally $n+2\delta_f n$ multiplications and divisions
where $\delta_f=0$  if $f=1$ and $\delta_f=1$ otherwise.


\subsection{Sparse ARSPH matrices based on Hadamard's processes}\label{shad}


{\bf FAMILY 5.}  A $2^k\times 2^k$ 
 {\em Abridged Recursive Scaled and Permuted Hadamard
 (ARSPH)
  matrix}  $H=H_{2^k,d}$ of {\em depth $d$}
 has $q=2^d$ nonzero entries 
in every row and in every column,  for a fixed integer $d$, $1\le d\le k$. Hence
 such a matrix is sparse unless 
$k-d$ is a small integer.

A special recursive structure 
of such a matrix 
allows
highly efficient  
 parallel implementation of  
its pre-multiplication by a vector 
(cf. 
Remark \ref{resprs0}).

We  recursively define matrices $H_{2^{h+1},d}$
  for $h=k-d,\dots,k-1$, as follows:

\begin{equation}\label{eqrfd}
H_{2^{h+1},d}=D_{2^{h+1}} P_{2^{h+1}}\begin{pmatrix}
H_{2^h,d} & H_{2^h,d} \\
H_{2^h,d} & -H_{2^h,d}
  \end{pmatrix}\bar P_{2^{h+1}}\bar D_{2^{h+1}},~H_{2^{k-d},d}=\begin{pmatrix}
I_{2^{k-d}} & I_{2^{k-d}}  \\
I_{2^{k-d}} & -I_{2^{k-d}} 
\end{pmatrix}.
\end{equation}
Here $P_{2^{h+1}}$ and $\bar P_{2^{h+1}}$ are $2^{h+1}\times 2^{h+1}$ permutation
matrices, $D_{2^{h+1}}$ and $\bar D_{2^{h+1}}$ are
$2^{h+1}\times 2^{h+1}$ diagonal 
matrices. 
We can pre-multiply a matrix $H_{2^k,d}$ by a vector by using at most $3dn$
flops.
 
If  the
matrices $D_{2^{h+1}}$ or $\bar D_{2^{h+1}}$ are real,  
having nonzero entries  $\pm 1$, 
then these flops  are
 additions and subtractions, and 
matrix $H_{2^k,d}$ is orthogonal up to scaling by a constant;
otherwise it is unitary.

For random  permutation matrices $P_i$ and $\bar P_i$
and the diagonal matrices $D_i$   and $\bar D_i$, the matrix $\widehat B=H_{2^k,d}$
depends on up to $(1+1/2+\dots+1/2^d)2^{k+2}=(1-1/2^d)2^{k+3}$ random  variables.

For $d=k$, the matrix  $H_{2^k,d}$ of (\ref{eqrfd})  turns into a dense (unabridged) RSPH matrix.

By letting $D_{2^{h}}=\bar D_{2^{h}}=I_{2^{h}}$,
  $P_{2^{h}}=\bar P_{2^{h}}=I_{2^{h}}$, or $D_{2^{h}}=\bar D_{2^{h}}=P_{2^{h}}=\bar P_{2^{h}}=I_{2^{h}}$
 for all $h$, we arrive at the three  sub-families  
{\em ARPH}, {\em ARSH}, and  {\em AH} of the family of ARSPH matrices.
For $d=k$, an AH matrix turns into the dense (unabridged) matrix
of {\em Walsh-Hadamard transform}.

Special sub-families of $2^k\times 2^k$
{\em Abridged Scaled and Permuted Fourier (ASPF)}
and 
{\em Abridged Scaled and Permuted Hadamard (ASPH)}
 matrices 
 use the same initialization of  (\ref{eqrfd}),
$$\Omega_{2^{k-d},d}=
H_{2^{k-d},d}=
\begin{pmatrix}
I_{2^{k-d}} & I_{2^{k-d}}  \\
I_{2^{k-d}} & -I_{2^{k-d}} 
\end{pmatrix},$$
 and are defined 
by the following recursive processes, 
which specialize (\ref{eqrfd}), 
\begin{equation}\label{eqfd}
\Omega_{2^{h+1},d}=
\widehat P_{2^{h+1}}\begin{pmatrix}\Omega_{2^h,d}&\\ &\Omega_{2^h,d}\end{pmatrix}\begin{pmatrix}I_{2^h}&\\ 
&\widehat D_{2^{h}}\end{pmatrix}\begin{pmatrix}I_{2^h}&~~I_{2^h}\\ I_{2^h}&-I_{2^h}\end{pmatrix},~
H_{2^{h+1},d}=\begin{pmatrix}
H_{2^h,d} & H_{2^h,d} \\
H_{2^h,d} & -H_{2^h,d}
  \end{pmatrix},
\end{equation}
for $h=k-d,\dots,k-1$ 
(cf. \cite[Section 2.3]{P01} and   \cite[Section 3.1]{M11}),
and output the matrices $P\Omega_{2^{k},d}D$ or $PH_{2^{k},d}D$
for fixed or random matrices $P$ and $D$ of  primitive types 1 and  2, respectively.
Here $\widehat D_{2^{h}}=\diag(\omega_{2^{h}}^i)_{i=0}^{2^h-1}$ and 
$\widehat P_{2^{h}}$ is the 
$2^{h}\times 2^{h}$ odd/even permutation matrix, such that 
$\widehat P_{2^{h}}({\bf u})={\bf v}$, ${\bf u}=(u_i)_{i=0}^{2^{h}-1}$, 
${\bf v}=(v_i)_{i=0}^{2^{h}-1}$, $v_j=u_{2j}$, $v_{j+2^{h-1}}=u_{2j+1}$, 
and $j=0,1,\ldots,2^{h-1}-1$.

The sub-families of ASPF and ASPH matrices  
in turn have sub-families of {\em ASF, APF,  
ASH, and APH} matrices.
The $n\times n$ AF  matrix for $d=k$ is just the matrix $DFT(n)$, $\Omega_n$ of Definition \ref{defdft}.

Recursive process (\ref{eqfd}) defining the matrix 
$\Omega_n$ is known as the decimation 
in frequency (DIF) radix-2 representation of FFT; transposition turns it into
the 
decimation 
in time (DIT) radix-2 representation  of FFT.  
The numbers of random variables involved  
into generation of general ARSPF and ARSPH  matrices 
 decrease to at most $(1-1/2^d)2^{k+1}$ for 
 ASPF and ASPH matrices, 
further  decrease to at most $(1-1/2^d)2^{k-1}$
for
ASF, APF,  
ASH, and APH matrices, and
the AF and AH matrices involve no random variables. 
The estimated arithmetic cost of  pre-mul\-ti\-pli\-ca\-tion of all these submatrices 
 by a vector is the same as in the case of 
ARSPF and ARSPH matrices.

Another well-known special case is given by
recursive
two-sided Partial Random Butterfly Transforms
(PRBTs),
 based on two unpublished Technical Reports of 1995 by Parker and by Parker and Pierce,
  but improved,  
carefully implemented, and then extensively tested  
in  \cite{BDHT13}.
 
For an $n\times n$ input matrix and  even $n=2q$, 
 that paper 
defines PRBT as follows,
\begin{equation}\label{eqprbt}
B^{(n)}=\frac{1}{\sqrt 2}\begin{pmatrix}
R  & S  \\
R  &  -S
\end{pmatrix}=\frac{1}{\sqrt 2}\begin{pmatrix}
I_{q}  & I_{q}  \\
I_q   &  -I_q 
\end{pmatrix} \diag(R,S)
\end{equation}
where $R$ and $S$ are random diagonal nonsingular matrices. 
The paper  \cite{BDHT13} defines multipliers $F$ and $H$  recursively
by using PRBT blocks. According to \cite{BDHT13},  
the two-sided recursive processes of  depth $d=2$ with PRBT blocks 
are ``sufficient in most cases". In such processes
$F=\diag (B_{1}^{(n/2)},B_{2}^{(n/2)})B^{(n)}$,
and the multiplier $H$ is defined similarly.
In the case of depth-$d$ recursion, $d\ge 2$, each of 
 the multipliers $F$ and $H$ is defined as the product
 of $d$ factors made up of $2^j$ diagonal blocks of size  $n/2^j\times n/2^j$,
for $j=0,\dots,d-1$, each block 
of the same type as above, and
the two-sided multiplication by a vector involves $6dn$ flops for $n=2^k$.


\subsection{Four additional families of sparse multipliers
}\label{scrcgv}


{\bf FAMILY 6.}  {\em Sparse} $f$-{\em circulant matrices} 
$H=Z_f({\bf v})$ are
 defined by a fixed or random scalar $f$, $|f|=1$, and 
the  first column ${\bf v}$ having exactly 
$q$ nonzero entries, for $q\ll n$.
The positions and values of non-zeros can be
 randomized (and then the matrix depends on up to $2n+1$ random  variables).

Such a matrix can be pre-multiplied by a vector by using at most 
$(2q-1)n$ flops  or, in the real case where $f=\pm 1$ and $v_i=\pm 1$
for all $i$, by using at most
$qn$ additions and subtractions. 

\medskip

{\bf FAMILY 7.} {Abridged \em $f^n$-circulant $n\times n$ matrices} $H$
of a recursion depth $d$ 
for a  scalar $f$, $n=2^k$, 
and two integers $d$ and $k$ 
 such that $|f|=1$,  $k>d\ge 0$,
  and the
integer $k-d$ is not small,
 $$H=(MD_f)^{-1}DMD_f.$$ Here
$D=\diag(d_i)_{i=1}^{n}$ and  $D_f=\diag(f^i)_{i=0}^{n-1}$ are diagonal matrices,
$|d_i|=1$ for all $i$ and $M$ is an 
AF or AH matrix of Family 5 of 
recursion depth $d$. 

Such a matrix $H$ is  unitary up to scaling by a constant (or  orthogonal
if the matrices $M$, $D$ and $D_f$ are real)
and can be pre-multiplied by a vector by using at most $6dn$ flops.
It involves up to $n$ random  variables, but would involve $3n$  (resp. 2n) such variables  
if we  use  an ASPF or ASPH (resp. ASF, ASH, ASF, or APH) 
rather than an AF or AH matrix $F$.
For $d=k$ the AF matrix $M$ turns into the DFT matrix $\Omega_n$
and 
$H$ turns into the
 $g$-circulant matrix (cf. Theorem \ref{thcpw}) 
$$Z_g({\bf v})=D_f^{-1}\Omega_n^HD\Omega_nD_f,~
{\rm for}~g=f^n,~D_f=\diag(f^i)_{i=0}^{n-1},
D=\diag(d_i)_{i=0}^{n-1},~{\rm and}~(d_i)_{i=0}^{n-1}=\Omega_nD_f{\bf v}.$$
For $f=1$, the above expressions are simplified: $g=1$, $D_f=I_n$, $M=\Omega_n$, and
$H=\sum_{i=0}^{n-1}v_iZ_1^i$
is a circulant matrix
$Z_1({\bf v})=\Omega_n^HD\Omega_n,~D=\diag(d_i)_{i=0}^{n-1},~{\rm for}~ 
(d_i)_{i=0}^{n-1}=\Omega_n{\bf v}.$

\medskip

{\bf FAMILY 8.} {\em Scaled  and permuted chains of Givens rotations with an AF or AH factor}.
A permutation matrix $P$  and $n-1$ angles $\theta_1,\dots,\theta_{n-1}$ 
together define a permuted chain 
of Givens rotations,  
$$G(\theta_1,\dots,\theta_{n-1})=P\prod_{i=1}^{n-1}G(i,i+1,\theta_i)~
{\rm for}~n=2^k.$$
One can combine two such chains $G_1$ and $G_2$
with scaling and DFT into the dense  unitary matrix
$$H=D_1G_1D_2G_2D_3\Omega_n$$ (cf. \cite[Section 4.6]{HMT11}),
for diagonal matrices
$D_1$, $D_2$ and $D_3$ of primitive type  2.

For $n=2^k\gg 2^d$ we can make   matrix $H$ sparse by replacing the factor $\Omega_n$
with  an $n\times n$ AF or AH matrix; then  we can 
 pre-multiply this  matrix by a vector by using $(1.5d+10)n$ flops.

A matrix $H$ involves up to $7n$ random  variables,
but  we can increase this bound by $2n$ (resp. by $n$) if we replace the factor $\Omega_n$ 
with an ASPF or ASPH
(resp. by APF, APH, ASF, or ASH) rather than AF or AH matrix. 
 
\medskip

{\bf FAMILY 9.}  {\em Pairs of scaled  and permuted Householder reflections  with an AF or AH factor.}
Define the unitary matrix $$H=\frac{1}{\sqrt n}D_1R_1D_2R_2D_3\Omega_n$$ 
where $D_1$,  $D_2$, and  $D_3$ are diagonal matrices,  $R_1$ and  $R_2$
are Householder reflections, and $\Omega_n$ is a DFT matrix.
This matrix can be multiplied by a vector by using about $7n$  flops
plus the cost of performing DFT. By randomizing diagonal, permutation and 
 Householder reflection matrices  
we can involve  up to $7n$ random  variables.
For $n=2^k$ and sparse vectors ${\bf h}_i$
defining Householder reflections $R_i$, for $i=1,2$,
we can make   matrix $H$ sparse by replacing the factor $\Omega_n$
with  an $n\times n$ AF or AH matrix of recursion depth $d$
 such that $2^d\ll n$; then  we can 
 multiply this  matrix by a vector by using at most $(1.5d+7)n$ flops.


\subsection{Estimated numbers of random  variables and flops}\label{sflprnd}


Table \ref{tabmlt} summarizes 
upper bounds on (a) the numbers of random variables involved into the matrices $\widehat B$
of Families 1--9
and (b) the numbers of flops for multiplication of such a matrix $\widehat B$ by
 a vector.\footnote{The matrices of Families 2--4 involve up to $7n-2$, $7n$, and $n$  random variables, respectively,
and are multiplied by a vector by using $O(n\log (n))$ flops.}
Compare these data with  $n^2$ random  variables
and $(2n-1)n$ flops involved in the case of  a  Gaussian $n\times n$  multiplier
and see more refined bounds in Sections \ref{sinvchh}--\ref{scrcgv}.


\begin{table}[ht] 
  \caption{The numbers of random  variables and flops}
\label{tabmlt}

  \begin{center}
    \begin{tabular}{|*{8}{c|}}
      \hline
Family & 1 & 5  &  6 &  7 & 8 &  9 
\\ \hline
random  variables &  $n-1$ & $2^{k+3}$ & $2q+1$  
& $n$ & $7n-2$  & $7n$ 
\\\hline
flops  & $2n-1$ &  $3dn$  & $(2q-1)n$     & $6dn$   & $(1.5d+10)n$
 & $(1.5d+7)n$
\\\hline

    \end{tabular}
  \end{center}
\end{table}

\begin{remark}\label{resprs0}
Other observations besides flop estimates  can be  decisive.
E. g., a
special recursive structure 
 of  an ARSPH matrix $H_{2^{k}}$ allows
highly efficient  
 parallel implementation of  
its multiplication by a vector based on 
Application Specific Integrated Circuits (ASICs) and 
Field-Programmable Gate Arrays (FPGAs), incorporating Butterfly
Circuits \cite{DE}.
\end{remark}


\subsection{Other basic families of multipliers}\label{sbscfml}


There is a variety of other interesting basic matrix families.
E.g., one can generalize Family 6  
to the family of sparse matrices with $q$ nonzeros entries $\pm 1$ 
in every row and in every column for a fixed integer $q$, 
$1\le q\ll n$. Such matrices
can be defined as the sums $\sum_{i=1}^q\widehat D_iP_i$
for fixed or random permutation matrices  $P_i$  and 
diagonal  matrices $\widehat D_i$,
 involve up to  $qn$ random  variables, and
 can be pre-multiplied by a vector at the same  estimated cost
as sparse $f$-circulant matrices.

For another example, one can modify a Givens chains of the form
$D_1G_1D_2G_2D_3F$, for $F$ denoting an DFT, AH,  ASPF, ASPH,
 APF, APH, ASF, or ASH matrix, by replacing one of the matrices $G_1$ or $G_2$
with a permuted Householder reflection matrix or the inverse of a bidiagonal matrix.

 The reader 
can find more families of multipliers
in our Section \ref{sexp1}.


\section{Symbolic and numerical GENP with one-sided randomized circulant 
pre-processing}\label{ssmblnm1s} 


\subsection{ GENP with one-sided randomized circulant pre-processing  
is likely to be safe universally (for any input)}\label{scrsmbl} 

 
In symbolic application of GENP one only cares about its safety rather than numerical safety.
In this case the power of randomization is strengthened.  
Claims (i) of Theorems \ref{thdgr} and \ref{thdgrd}
imply that GENP pre-processed with  any  nonsingular multiplier $F$ or $H$
(e.g.,  $H=I_n$) 
or with any  pair
of such  multipliers
is unsafe only for an algebraic variety 
of lower dimension in the space of all inputs 
and that for a fixed input  GENP is safe 
if  pre-processed with almost any  nonsingular multiplier
 or any pair of such  multipliers
apart from  an algebraic variety 
of lower dimension.
Therefore in symbolic computations 
one can quite confidently apply GENP 
with no pre-processing and in the very unlikely case of 
failure re-apply GENP with a  random nonsingular multiplier. 

For a theoretical challenge, however,
one can seek 
randomized multipliers that support safe GENP and BGE universally, that is, with probability 1
for any nonsingular input. This challenge has been met
 already in 1991 (see, e.g.,  \cite[Section 2.13,  
entitled ``Regularization of a Matrix via 
Preconditioning with Randomization"]{BP94}). 
Among the known options, 
 one-sided pre-processing with random Toeplitz  multipliers of \cite{KP91} 
is most efficient,
but a little inferior two-sided 
 pre-processing with random triangular Toeplitz  multipliers
of \cite{KS91} has been more widely advertised in the Computer Algebra community  and has become more popular
there.

Next we prove 
that even pre-processing with one-sided Gaussian or
random uniform circulant multipliers
(see Section  \ref{stplcrc} for definitions) is likely to make
GENP safe, that is,
 involving no divisions by 0.
Using circulant multipliers  saves 50\% of random  variables and
enables a 4-fold (resp. 2-fold) acceleration of the pre-processing of \cite{KS91}
(resp.  \cite{KP91}).

We need more than two pages besides the space for definitions 
in order to prove that result,
but it  
enables us to improve substantially
the popular and decades-old recipes 
for pre-processing GENP for
symbolic computations.
Namely, pre-processing  of \cite{KS91} requires 
pre- and post-multiplication of an $n\times n$ input matrix $A$ by 
an upper and a lower triangular Toeplitz matrices, respectively,
at the overall cost dominated by the cost of performing twelve DFT$(n)$
per row of an input matrix $A$
(see Remark \ref{rettm}),
and in addition one must 
 generate $2n-1$ random variables. 
Pre-processing  of \cite{KP91}
uses as many random variables
and six DFT$(n)$
per row of $A$. 
Our present algorithm only
post-  or pre-multiplies a matrix $A$ by a
single circulant matrix, at the cost 
 dominated by the cost of performing
three  DFT$(n)$
per row of an input matrix $A$,
 and uses only  $n$ random  variables.

Let us supply the details.

\begin{theorem}\label{th1} 
Suppose $A=(a_{i,j})^n_{i,j=1}$ is a nonsingular matrix,  
 $T=(t_{i-j+1})^{n}_{i,j=1}$ is a Gaussian $f$-circulant matrix, 
$B=AT=(b_{i,j})^{n}_{i,j=1}$, 
 $f$ is a fixed complex number, $t_{1},\dots, t_{n}$ are  variables, 
and  $t_{k}=ft_{n+k}$ for $k=0, -1, \dots, 1-n$.
Let $B_{l,l}$ denotes the $l$-th leading blocks of the matrix $B$
 for $l=1,\dots,n$, and so $\det(B_{l,l})$ are  
 polynomial in $t_{1}, \dots, t_{n}$, for all $l$, 
$l = 1,\dots,n$.
Then neither of these polynomials vanishes identically in  $t_{1}, \dots, t_{n}$.
\end{theorem}

\begin{proof}
 Fix a positive integer $l\le n$.
With the convention $\alpha_{k\pm n} = f\alpha_{k}$, for $k = 1,\dots,n$,  we can write
\begin{equation}\label{DetExpan}
B_{l,l}=
\Big (\sum^{n}_{k_1=1}\alpha_{k_1}t_{k_1},\sum^{n}_{k_2=1}\alpha_{k_2+1}t_{k_2},\dots,
\sum^{n}_{k_{l}=1}\alpha_{k_{l}+l-1}t_{k_{l}} \Big ),
\end{equation}
where $\alpha_{j}$ is the $j$th column of  $A_{l,n}$. 
Let $a_{i, j+n} = fa_{i, j}$, for $k = 1,\dots,n$, and readily  verify that
$$
b_{i,j}=\sum^{n}_{k=1}a_{i,j+k-1}t_{k},$$ 
and so $\det(B_l)$ is a homogeneous polynomial in $t_{1}, \dots, t_{n}$.

Now  Theorem \ref{th1} is implied by the following lemma.

\begin{lemma}\label{le0}
If  $\det(B_{l,l})=0$ identically in all the variables
 $t_{1}, \dots, t_{n}$, then 
\begin{equation}\label{Detalpha}
 \det(\alpha_{i_{1}}, \alpha_{i_{2}}, \dots, \alpha_{i_{l}}) = 0
\end{equation}
 for all 
$l$-tuples of  subscripts $(i_{1}, \dots, i_{l})$
such that $1\le i_{1}< i_{2}< \cdots <  i_{l}\le n$.
\end{lemma}

Indeed let $A_{l,n}$ denote the block submatrix made up of the first $l$ rows of $A$.
Notice that  
if (\ref{Detalpha}) holds for all $l$-tuples 
of the subscripts $(i_{1},\dots, i_{l})$ above, 
then the rows of the block  submatrix  $A_{l,n}$ are linearly dependent,
but they are 
the rows of the matrix  $A$, and their linear dependence
 contradicts the assumption that the matrix $A$ is nonsingular.

In the rest of this section we prove Lemma \ref{le0}.
At first we order the $l$-tuples 
$I=(i_{1},\dots, i_{l})$, each made up of 
$l$ positive integers written in nondecreasing order,
 and then we apply induction.

We order all $l$-tuples of integers by ordering at first
their largest integers, in the case of ties by
 ordering
their second largest integers, and so on.
 
We can define the classes of these  $l$-tuples
up to permutation of their integers and congruence modulo $n$, and then represent every class 
by the  $l$-tuple of nondecreasing integers between $1$ and $n$. 
Then our ordering of  $l$-tuples
of ordered integers takes the following form,
 $(i_{1},\dots,i_{l})<(i'_{1},\dots,i'_{l})$ if and only if there exist a subscript $j$ such that $i_{j}<i'_{j}$ and $i_{k}=i'_{k}$ for $k= j+1, \dots, l$.

We begin our proof of Lemma \ref{le0} with the following basic result.
 
\begin{lemma}\label{le1}
It holds that
 $$\det (B_{l,l})=\sum_{1\le i_{1}\le i_{2}\le \dots \le  i_{l}\le n} a_{\prod^{l}_{j=1}t_{i_{j}}}\prod^{l}_{j=1}t_{i_{j}}$$  where a tuple $(i_{1},\dots,i_{l})$ 
may contain repeated elements,
\begin{equation}\label{DetCoeff}
a_{\prod^{l}_{j=1}t_{i_{j}}}=\sum_{(i'_{1},\dots,i'_{l})}\det(\alpha_{i'_{1}},\alpha_{i'_{2}+1}, \dots,\alpha_{i'_{l}+l-1}),
\end{equation}
and $(i'_{1},\dots,i'_{l})$ ranges over all permutations of $(i_{1},\dots,i_{l})$.
\end{lemma}

\begin{proof}
By using (\ref{DetExpan})  we can expand $\det(B_{l,l}) $ as follows,

\begin{align}
\det(B_{l,l}) &= 
\det{\Big (\sum^{n}_{k_1=1}\alpha_{k_1}t_{k_1},\sum^{n}_{k_2=1}\alpha_{k_2+1}t_{k_2},\dots,\sum^{n}_{k_l=1}\alpha_{k_l+l-1}t_{k_l}\Big )} \nonumber \\
&= \sum^{n}_{i_{1}=1} t_{i_{1}}\det{\Big (\alpha_{i_{1}}, \sum^{n}_{k=1}\alpha_{k+1}t_{k},\dots,\sum^{n}_{k_l=1}\alpha_{k_l+l-1}t_{k_l}\Big )} \nonumber \\
&= \sum^{n}_{i_{1}=1} t_{i_{1}}\sum^{n}_{i_{2}=1}t_{i_{2}}\det{\Big (\alpha_{i_{1}}, \alpha_{i_{2}+1}, \sum^{n}_{k_2=1}\alpha_{k_2+2}t_{k_2}, \dots, \sum^{n}_{k_l=1}\alpha_{k_l+l-1}t_{k_l}\Big )} 
\nonumber \\
&
= \dots \nonumber \\
&= \sum^{n}_{i_{1}=1} t_{i_{1}}\sum^{n}_{i_{2}=1}t_{i_{2}}\dots\sum^{n}_{i_{l}=1}t_{i_{l}}\det{(\alpha_{i_{1}}, \alpha_{i_{2}+1}, \dots, \alpha_{i_{l}+l-1})}.
\end{align}

Consequently the coefficient 
$a_{\prod^{l}_{j=1}t_{i_{j}}}$ of  any  term $\prod^{l}_{j=1}t_{i_{j}}$
is the sum of all determinants $$\det{(\alpha_{i'_{1}}, \alpha_{i'_{2}+1}, \dots, \alpha_{i'_{l}+l-1})}$$ where $(i'_{1},\dots,i'_{l})$ ranges over all permutations of $(i_{1},\dots,i_{l})$,  and we arrive at (\ref{DetCoeff}).

\end{proof}

In particular, the coefficient of the term $t_{1}^{l}$ is $a_{t_{1}\cdot t_{1}\cdot \cdots \cdot t_{1}} = \det(\alpha_{1}, \alpha_{2}, \dots, \alpha_{l})$. This coefficient equals zero because $B_{l, l}$ is identically zero, by assumption
 of lemma \ref{le0},
 and we obtain  
\begin{equation}\label{Det0}
 \det(\alpha_{1}, \alpha_{2}, \dots, \alpha_{l})=0.
\end{equation} 
This is the basis of our inductive proof of Lemma \ref{le0}.
In order to complete the induction step, it remains to prove the following lemma.

\begin{lemma}\label{DetInd}
Let $J=(i_{1},\dots,i_{l})$ be a tuple such that $1\le i_{1}< i_{2}<\cdots< i_{l}\le n$. 

Then $J$ is a subscript tuple
 of the coefficient of the term $\prod^{l}_{j=1}t_{i_{j}-j+1}$
 in equation (\ref{DetCoeff}). \

Moreover, $J$ is the single largest tuple among all subscript tuples.
\end{lemma}

\begin{proof}
Hereafter $\det(\alpha_{i'_{1}}, \alpha_{i'_{2}+1}, \dots, \alpha_{i'_{l}+l-1})$ is said to be 
\textit{the determinant associated with the permutation}
$(i'_{1}, \dots, i'_{l})$  of $(i_{1}, \dots, i_{l})$ 
in (\ref{DetCoeff}). Observe that
$\det{(\alpha_{i_{1}}, \dots, \alpha_{i_{l}})}$ is the determinant associated with 
$\mathcal I=(i_{1}, i_{2}-1, \dots, i_{l}-l+1)$ in the coefficient 
$a_{\prod^{l}_{j=1}t_{i_{j}-j+1}}$. 

Let $\mathcal I'$ 
be a permutation of $\mathcal I$. Then $\mathcal I'$ can be written as 
$\mathcal I' = (i_{s_{1}}-s_{1}+1, i_{s_{2}}-s_{2}+1, \dots, i_{s_{l}}-s_{l}+1)$, 
where $(s_{1}, \dots, s_{l})$ is a permutation of $(1, \dots, l)$.
 The determinant associated with $\mathcal I'$ has the subscript tuple 
$\mathcal J' = (i_{s_{1}}-s_{1}+1, i_{s_{2}}-s_{2}+2, \dots, i_{s_{l}}-s_{l}+l)$. 
$j$ satisfies the inequality
$j\leq i_j\leq n-l+j$
because by assumption $1\leq i_1< i_2 <\dots< i_l\leq n$, for any $j=1, 2, \dots, l$. 
Thus, $i_{s_j} - s_j+j$ satisfies the inequality 
$j\leq i_{s_j} - s_j + j\leq n-l+j\leq n$,
 for any $s_j$. This fact implies that no subscript of $\mathcal I'$ is
 negative or greater than n.

Let $\mathcal J''= 
(i_{s_{r_{1}}}-s_{r_{1}}+r_{1},i_{s_{r_{2}}}-s_{r_{2}}+r_{2}, 
\dots,i_{s_{r_{l}}}-s_{r_{l}}+r_{l})$ 
be a permutation of $\mathcal J$ such that its elements 
are arranged in the nondecreasing order. 
Now suppose $\mathcal J'' \geq J$.
Then we must have 
$i_{s_{r_{l}}}-s_{r_{l}}+r_{l} \geq i_{l}$.
This implies that

\begin{equation}\label{eqil}
 i_{l}-i_{s_{r_{l}}} \leq r_{l} - s_{r_{l}}.
\end{equation}
Observe that 
\begin{equation}\label{eqilsl}
l - s_{r_{l}} \leq
 i_{l}-i_{s_{r_{l}}}
\end{equation}
 because 
$i_1<i_2<\dots<i_l$ by assumption.
Combine bounds (\ref{eqil}) and  (\ref{eqilsl})
and obtain that 
$l - s_{r_{l}} \leq
 i_{l}-i_{s_{r_{l}}} \leq r_{l} - s_{r_{l}}$
and hence $r_{l}=l$.

 Apply this argument recursively for $l-1, \dots, 1$
and obtain that $r_{j} = j$ for any $j = 1, \dots, l$.
 Therefore  $\mathcal J = \mathcal J'$ and $\mathcal I' = \mathcal I$.
It follows that $J$ is indeed the single largest subscript tuple.
\end{proof}

By combining Lemmas \ref{le1} and  \ref{DetInd}, we support the induction step
of the proof of  Lemma \ref{le0}, which we summarize as follows:

\begin{lemma}\label{le2}  
Assume the class of $l$-tuples of $l$ positive integers
written in the increasing order in each $l$-tuple and
write $\det(I)=\det(\alpha_{i_{1}},\alpha_{i_{2}}, \dots,\alpha_{i_{l}})$
if $I=(\alpha_{i_{1}},\alpha_{i_{2}}, \dots,\alpha_{i_{l}})$.  

Then 
$\det (I)=0$ provided that $\det (J)=0$ for all $J< I$.
\end{lemma}

Finally we readily deduce  
 Lemma \ref{le1} by
combining this result with equation (\ref{Det0}).
This completes the proof of Theorem \ref{th1}.
\end{proof}


\begin{corollary}\label{co0}
Assume  
 any nonsingular $n\times n$ matrix $A$ and a finite set $\mathcal S$ of cardinality $|\mathcal S|$.
Sample the values of the $n$ coordinates $v_1,\dots,v_n$
of a vector ${\bf v}$ at random from this set.
Fix a complex $f$ and define the 
matrix 
$H=Z_f({\bf v})$ of size $n\times n$, with the
first column vector ${\bf v}=(v_i)_{i=1}^n$. 
  Then 
GENP and BGE are safe for the matrix $AH$

(i) with a probability of at least $1-0.5(n-1)n/|\mathcal S|$
if the values of the $n$ coordinates $v_1,\dots,v_n$
of a vector ${\bf v}$ have been sampled uniformly at random from a finite  set
 $\mathcal S$ of cardinality $|\mathcal S|$
or

(ii) with probability 1 if these  coordinates are i.i.d. Gaussian  variables.

(iii) The same claims hold for the matrix $FA$.
\end{corollary}


\begin{proof}
Theorems \ref{th0}, \ref{th1}, and \ref{thrnd} together imply
claims (i) and (ii) of  the corollary.
By applying transposition to the matrix $AH$, extend the results to claim (iii).
\end{proof}


\subsection{GENP with any one-sided circulant pre-processing
fails numerically for some specific inputs}\label{scrfl} 

 
By virtue of Corollary \ref{co0} random circulant pre-processing is a
universal means for ensuring safe GENP, but is it also 
a universal means for ensuring numerically safe  GENP?

By virtue of \cite[Corollary 6.3.1]{P16},
the answer is ``No",\footnote{This nontrivial result has been deduced by using the recent specialization in \cite{P15}
to Cauchy and Van\-der\-monde matrices of 
the general techniques of transformation of matrix structure proposed in \cite{P90}.}
and moreover
 GENP is numerically unsafe when it is applied to
the DFT matrix $\Omega_n$ already for a reasonably large integer $n$
as well as to the matrices $\Omega_nZ_1({\bf v})$
and any vector ${\bf v}$ of dimension $n$, that is, for
any circulant matrix $Z_1({\bf v})$, and consequently for a random circulant matrix
$Z_1({\bf v})$.
By combining the proof of  \cite[Corollary 6.3.1]{P16} with Theorem \ref{thcpw}
one can immediately extend the result to any $f$-circulant 
multiplier $Z_f({\bf v})$ for any $f\neq 0$.
It follows that GENP also fails numerically for the input pairs 
$(A,H)$ where $A=\Omega_nQ$ and $H=Q^HZ_f({\bf v})$
for any unitary matrix $Q$.

Surely one does not need to use GENP in order to solve a
linear system of equations with a DFT coefficient matrix,
but the above results reveal the difficulty in finding universal
classes of structured pre-processing for GENP.
Having specific bad pairs of inputs and multipliers
does not contradict claim (ii) of Corollary \ref{cogpr}, and actually in extensive
tests in \cite{PQZ13} and \cite{PQY15} very good numerical stability 
 has been observed when we applied GENP to
 various classes of nonsingular well-conditioned 
input matrices with random circulant multipliers.


\section{Alternatives: 
augmentation and additive pre-processing}\label{saltaug}


Other randomization techniques, besides
multiplications, are also beneficial for various 
fundamental matrix computations (see \cite{M11},
 \cite{HMT11},  \cite{PGMQ}, \cite{PIMR10}, \cite{PQ10},  \cite{PQZC}, \cite{PQ12},
 \cite{PZa}, and the bibliography therein). 
By combining our present results with those of \cite{PZa}, we next 
supports GENP and BGE 
by means of  
 {\em randomized augmentation} of a matrix,
that is, of appending to it random rows and columns,
as well as by means of an alternative  and closely related
technique of {\em additive pre-processing}
(cf. (\ref{eqkc})--(\ref{eqkppkc})).   


\subsection{Gaussian augmentation}\label{sgaug} 


By virtue of claim (i) of 
\cite[Theorem 10.1]{PZa},
properly scaled Gaussian augmentation 
 of a sufficiently large size
is likely to produce
strongly nonsingular and strongly well-conditioned
matrices.
Namely, this holds when we 
augment an $n\times n$ matrix $A$ and 
produce the matrix $K=\begin{pmatrix}
I_h  &   V^T   \\
U   &   A
\end{pmatrix}.$ 
Here $U$ and $V$ are 
$n\times h$ Gaussian matrices
filled with
 $2hn$ i.i.d.  Gaussian random entries
and $\nu\le h\le n$, for 
$\nu$ denoting an upper bound on the numerical nullity
 of the leading blocks,
that is, on their numerical co-rank.
In the dual version of that theorem, average matrix $K$ is strongly nonsingular and strongly well-conditioned
if $A$ is a Gaussian matrix and if the matrices $U$ and $V$
 have full rank and 
are scaled so that $||U||\approx ||V||\approx 1$,
are well-conditioned.


\subsection{SRFT augmentation}\label{ssrftaug} 


By virtue of  \cite[Section 8]{PZa}  
we are likely to succeed   if
 we apply GENP to the matrix $K$ above
obtained by augmenting a nonsingular and  well-conditioned
matrix $A$ with SRFT matrices $U$ and $V$
of a sufficiently large size replacing the Gaussian ones
of the previous subsection.\footnote{The same results and  estimates hold if one substitutes 
matrices of SRHT for the  ones of SRFT
(cf. \cite{T11}).}

Let us supply some details.
Claim (ii) of \cite[Theorem 10.1]{PZa} implies
that we are likely to produce a 
strongly nonsingular and 
strongly well-conditioned matrix $K$
if $U$ and $V$ are SRFT matrices such that
$\nu\ge q=cn$ for a sufficiently large constant $c$ 
and if
\begin{equation}\label{eqsrft}
4\Big(\sqrt {\nu}+\sqrt {8\log_2(\nu n)}\Big)^2\log_2(\nu)\le h.
\end{equation}

Indeed, the $q\times q$ leading blocks $K_{q,q}$ 
of the matrix $K$ are the identity
matrices $I_q$ for $q=1,\dots, h$, and so 
we only need to estimate the probability that 
the  leading blocks $K_{q,q}$ are 
 well-conditioned for all $q>h$
because their
nonsingularity with probability 1 readily follows from Theorem \ref{thrnd}.

It is proven in  \cite[Section 8]{PZa}
that under (\ref{eqsrft}) this property holds with a probability at least $1-c'/q$
for a constant $c'$ and a fixed $q$.
Therefore it holds for all $q$, $q=h+1,\dots,h+n$
with  a probability at least
$1-c'\sum_{q=h+1}^{h+n}1/q\approx 1-c'\ln (\frac{h+n}{h+1})=
1-c'\ln(1+\frac{n-1}{h+1})\ge 1-c' \ln(1+1/c)$.
This is close to 1 for a sufficiently large constant $c$,
and our claim about the matrix $K$ follows.


\subsection{Linking augmentation to additive pre-processing}\label{saugadd} 


Consider an augmented matrix $K$
and its inverse $K^{-1}$.
Then  its $n\times n$ trailing (that is, southeastern)
block is $C^{-1}$ for $C=A-UV^T$.
Indeed 

\begin{equation}\label{eqkc}
\begin{aligned}
K=\begin{pmatrix}
I_h  &  O_{h,n}  \\
U  & I_n
\end{pmatrix}
\begin{pmatrix}
I_h  &  O_{h,n} \\
O_{n,h}  &  C
\end{pmatrix}
\begin{pmatrix}
I_h  &  V^T  \\
O_{n,h}  & I_n
\end{pmatrix}.
\end{aligned}
\end{equation}
Consequently
\begin{equation}\label{eqkc-}
\begin{aligned}
K^{-1}=\begin{pmatrix}
I_h  &  -V^T  \\
O_{n,h}  & I_n
\end{pmatrix}
\begin{pmatrix}
I_h  &  O_{h,n} \\
O_{n,h}  &  C^{-1}
\end{pmatrix}
\begin{pmatrix}
I_h  &  O_{h,n}  \\
-U  & I_n
\end{pmatrix}
\end{aligned},
\end{equation}
$C^{-1}=\diag(O_{h,h}I_n)K^{-1}\diag(O_{h,h}I_n)$,
and the claim follows.
 
Now deduce that 
$||K^{-1}||/N\le ||C^{-1}||\le ||K^{-1}||$,
$||K||/N\le ||C^{-1}||\le N~||K||$,
and hence 
\begin{equation}\label{eqkppkc}
\kappa (K)/N^2\le \kappa (C)\le N\kappa (K),
~{\rm for}~ N=(1+||U||)(1+||V||).
\end{equation}
Equations 
(\ref{eqkc})-(\ref{eqkppkc}) 
closely link
augmentation $A\rightarrow K$ with {\em additive pre-processing}
$A\rightarrow C$ and also
link the leading blocks of  the augmented matrices
with those output by 
 additive pre-processing. 

Indeed readily extend our observations
to obtain that the $k\times k$
trailing submatrix of the leading block
$K_{h+k,h+k}^{-1}$ of the matrix $K$
is the $k\times k$ leading block $C_{k,k}^{-1}$
of the matrix $C^{-1}$
and that
\begin{equation}\label{eqkppkc1}
\kappa (K_{h+k,h+k})/N^2\le \kappa (C_{k,k})\le N\kappa (K_{h+k,h+k}),
\end{equation}
 for $k=1,\dots,n$.
If the factor $N$ is  reasonably bounded,
which is likely to be the case for
Gaussian matrices $U$ and $V$,  
 then
 the matrix $C_{k,k}$ is nonsingular and well-conditioned 
if and only if so is the matrix $K_{h+k,h+k}$.


\subsection{Transition back to computations with the original matrix.
Expansion  and
homotopy continuation}\label{ssmwhc} 


Having computed the inverses $K^{-1}$
and $C^{-1}$ by
applying GENP or BGE
to the augmented  matrix $K$,
one can
simplify computation of the inverse $A^{-1}$ of
the original matrix $A$ by applying the
 Sherman--\-Mor\-rison--\-Wood\-bury formula\footnote{Hereafter we use the acronym {\em SMW}.}
\begin{equation}\label{eqsmw}
A^{-1}=C^{-1}-C^{-1}U(I_h+V^TC^{-1}U)^{-1}V^TC^{-1},
\end{equation}
for $C=A-UV^T$ (cf. \cite[page 65]{GL13}).

Computing the inverse $A^{-1}$ by means of SMW
formula (\ref{eqsmw}) may still cause numerical problems
at the stages of computing and inverting the matrix
$I_h+V^TC^{-1}U$, but they are
less likely to occur if the matrix $C$ is nonsingular
and well-conditioned, which we expect to be the case in
this application.

In order to strengthen the chances for the success of this approach
we can apply some heuristic recipes. In our tests with benchmark inputs,
we succeeded by simply doubling the lower bound $\nu$
on the dimension $h$ of additive pre-processing
or equivalently  by keeping the same bound $\nu$ but requiring that $2\nu\le h\le n$.
Another natural remedy is the well-known  general technique of {\em homotopy continuation}, 
with which one proceeds as follows.
Fix two matrices $U$ and $V$ as before  and define the matrices 
$C(\tau)=A-\tau UV^T$ and 
$S(\tau)=I_h+\tau V^TC(\tau)^{-1}U$
for a nonnegative parameter $\tau$.
Suppose that the matrix $S(\bar \tau)$ is diagonally dominant
for some positive $\bar \tau$.
Notice that  the matrices  $C(1)=A-UV^T$ 
and $S(0)=I_h$ are readily invertible,
 fix the decreasing sequence of the values 
$\tau_k$, $k=0,1,\dots,l$, such that $\tau_0=1>\tau_1>\cdots>\tau_l=0$,
and compute the sequence of the matrices 
$\tau_j V^TC(\tau_j)^{-1}U$, for $j=0,1,\dots,l$,
by extending SMW formula 
 (\ref{eqsmw}) as follows,
$$C(\tau_{j+1})^{-1}=
C(\tau_{j})^{-1}-\Delta_j C(\tau_{j})^{-1}U(I_h+\Delta_j V^TC(\tau_{j})^{-1}U)^{-1}V^TC(\tau_{j})^{-1},$$
for $\Delta_j =\tau_{j+1}-\tau_{j}$.
For sufficiently small values $\Delta_j$, the matrices
$I_h+\Delta_j V^TC(\tau_{j})^{-1}U$ are diagonally dominant and readily invertible,
and then we can numerically safely perform $l$ homotopy 
 continuation steps for the transition from the inverse $C(1)^{-1}=C^{-1}$
to $C(0)^{-1}=A^{-1}$.

By generalizing
 SMW formula 
(\ref{eqsmw}) and writing $C=A-UV^T$, we can  readily  
express the inverse $K^{-1}$ 
of the augmented matrix $K$,
$$K^{-1}=\begin{pmatrix}
I_{h}  &  -V^TC^{-1}\\
O_{n,h} &  C^{-1}\end{pmatrix}
\begin{pmatrix}I_{h}  &  O_{h,n}\\
-U  &  I_n\end{pmatrix},~
K=\begin{pmatrix}
I_{h}  &  V^T\\
U  &  A\end{pmatrix}=\begin{pmatrix}
I_{h}  &  O_{h,n}\\
U  &  I_n\end{pmatrix} \begin{pmatrix}
I_{h}  &  V^T\\
O_{n,h} &  C\end{pmatrix}.$$



\section{Numerical Experiments}\label{sexp1}


 Numerical experiments have been
 performed,  by using MATLAB, 
by the second author
 in the Graduate Center of the City University of New York 
on a Dell computer with the Intel Core 2 2.50 GHz processor and 4G memory running 
Windows 7. 
Gaussian matrices have been  generated by applying the standard normal distribution function randn of MATLAB.
We refer the reader to \cite{PQZ13}, 
\cite{PQY15}, and 
 \cite{PZa}
for other extensive tests of 
 GENP with randomized pre-processing.




Tables \ref{tab63}--\ref{tab65a} show the maximum, minimum and average relative residual norms 
$||A{\bf y}-{\bf b}||/||{\bf b}||$ as well as the standard deviation
for the solution of 1000 linear system $A{\bf x}={\bf b}$ with  
Gaussian vector ${\bf b}$ and $n\times n$
input matrix $A$ for each $n$, $n=256, 512, 1024$ and $10$ linear systems for $n=2048, 4096$,
\begin{equation}\label{eqtsts}
 A=\begin{pmatrix}
A_k  &  B  \\
C    &  D
\end{pmatrix},
\end{equation}
with $k\times k$ blocks $A_k$, $B$, $C$ and $D$, for $k=n/2$,   
scaled so that
 $||B||\approx ||C||\approx ||D||\approx 1$,
 the  $k-4$ singular values of the matrix $A_k$ 
were equal 1 and the other ones were set to 0 (cf. \cite[Section 28.3]{H02}),
and with Gaussian Toeplitz matrices $B$, $C$, and $D$,
that is, with Toeplitz matrices of (\ref{eqtz}), each
defined by the i.i.d. Gaussian entries of its first row and first column.
(The norm  $||A^{-1}||$ 
 ranged from $2.2\times 10^1$ to  
$3.8\times 10^6$ in these tests.)
The linear systems have been solved by using GEPP, GENP, 
or GENP pre-processed with real Gaussian, real Gaussian  circulant, 
and  random  circulant 
multipliers, each
followed by a single loop 
 of iterative refinement.

In the tests covered in Table \ref{tab65a}, the matrix
$A$  was set to equal $\Omega$, the matrix of DFT$(n)$. For pre-processing,
 either Gaussian or Gaussian  unitary circulant matrices
$C=\Omega^{-1}D(\Omega{\bf v})\Omega$ have been used
as multipliers,
with ${\bf v}=(v_i)_{i=0}^{n-1}$, $v_i=\exp(2\pi\phi_i\sqrt {-1}/n)$ and
$n$ i.i.d. 
real Gaussian variables $\phi_i$,
$i=0,\dots,n-1$ (cf. Theorem \ref{thcpw}
and  Remark \ref{reblck}).  

As should be expected, GEPP has always produced accurate solutions,
with the average relative residual norms ranging from $10^{-12}$ 
to $7\times 10^{-13}$,  
but GENP with no pre-processing has consistently 
produced corrupted output with relative residual norms 
 ranging  from
$10^{-3}$ to $10^2$ for the input matrices $A$ of 
equation (\ref{eqtsts}). Even much worse was the output accuracy
when GENP 
with no pre-processing
or with Gaussian circulant pre-processing
was applied to the matrix $A=\Omega$.
In all other cases, however, GENP with   
random circulant pre-processing 
and 
with a single loop of iterative refinement
has produced solution with desired accuracy,
matching the output accuracy of GEPP.
Furthermore GENP has performed  similarly when 
it was applied to a nonsingular and well-conditioned input
pre-processed with a Gaussian multiplier.



\begin{table}[ht]
  \caption{Relative residual norms:
 GENP with Gaussian multipliers}
  \label{tab63}
  \begin{center}
    \begin{tabular}{| c | c | c | c | c | c | c |}
      \hline
      \bf{dimension} & \bf{iterations} & \bf{mean} & \bf{max} & \bf{min} & \bf{std} \\ \hline
 $256$ & $0$ & $6.13\times 10^{-9}$ & $3.39\times 10^{-6}$ & $2.47\times 10^{-12}$ & $1.15\times 10^{-7}$ \\ \hline
 $256$ & $1$ & $3.64\times 10^{-14}$ & $4.32\times 10^{-12}$ & $1.91\times 10^{-15}$ & $2.17\times 10^{-13}$ \\ \hline
 $512$ & $0$ & $5.57\times 10^{-8}$ & $1.44\times 10^{-5}$ & $1.29\times 10^{-11}$ & $7.59\times 10^{-7}$ \\ \hline
 $512$ & $1$ & $7.36\times 10^{-13}$ & $1.92\times 10^{-10}$ & $3.32\times 10^{-15}$ & $1.07\times 10^{-11}$ \\ \hline
 $1024$ & $0$ & $2.58\times 10^{-7}$ & $2.17\times 10^{-4}$ & $4.66\times 10^{-11}$ & $6.86\times 10^{-6}$ \\ \hline
 $1024$ & $1$ & $7.53\times 10^{-12}$ & $7.31\times 10^{-9}$ & $6.75\times 10^{-15}$ & $2.31\times 10^{-10}$ \\ \hline
 $2048$ & $0$ & $4.14\times 10^{-9}$ & $8.16\times 10^{-9}$ & $8.27\times 10^{-10}$ & $3.72\times 10^{-9}$\\ \hline
  $2048$ &      $1$ & $7.61\times 10^{-12}$ & $1.08\times 10^{-11}$ & $3.27\times 10^{-12}$ & $3.89\times 10^{-12}$\\ \hline
  $4096$ &      $0$ & $5.02\times 10^{-7}$ & $1.23\times 10^{-6}$ & $1.14\times 10^{-7}$ & $6.29\times 10^{-7}$\\ \hline
  $4096$ &      $1$ & $5.44\times 10^{-11}$ & $1.53\times 10^{-10}$ & $2.64\times 10^{-12}$ & $8.52\times 10^{-11}$\\ \hline

    \end{tabular}
  \end{center}
\end{table}



\begin{table}[ht]
  \caption{Relative residual norms:
 GENP with Gaussian circulant
  multipliers}
  \label{tab64}
  \begin{center}
    \begin{tabular}{| c | c | c | c | c | c | c |}
      \hline
      \bf{dimension} & \bf{iterations} & \bf{mean} & \bf{max} & \bf{min} & \bf{std} \\ \hline
 $256$ & $0$ & $8.97\times 10^{-11}$ & $1.19\times 10^{-8}$ & $6.23\times 10^{-13}$ & $4.85\times 10^{-10}$ \\ \hline
 $256$ & $1$ & $2.88\times 10^{-14}$ & $2.89\times 10^{-12}$ & $1.89\times 10^{-15}$ & $1.32\times 10^{-13}$ \\ \hline
 $512$ & $0$ & $4.12\times 10^{-10}$ & $3.85\times 10^{-8}$ & $2.37\times 10^{-12}$ & $2.27\times 10^{-9}$ \\ \hline
 $512$ & $1$ & $5.24\times 10^{-14}$ & $5.12\times 10^{-12}$ & $2.95\times 10^{-15}$ & $2.32\times 10^{-13}$ \\ \hline
 $1024$ & $0$ & $1.03\times 10^{-8}$ & $5.80\times 10^{-6}$ & $1.09\times 10^{-11}$ & $1.93\times 10^{-7}$ \\ \hline
 $1024$ & $1$ & $1.46\times 10^{-13}$ & $4.80\times 10^{-11}$ & $6.94\times 10^{-15}$ & $1.60\times 10^{-12}$ \\ \hline    
   $2048$ 	 &      $0$ 	 & $1.03\times 10^{-8}$ 	 & $2.87\times 10^{-8}$ 	 & $1.13\times 10^{-9}$ 	 & $1.59\times 10^{-8}$\\ \hline   
  $2048$ 	 &      $1$ 	 & $3.74\times 10^{-13}$ 	 & $6.09\times 10^{-13}$ 	 & $9.69\times 10^{-14}$ 	 & $2.59\times 10^{-13}$\\ \hline   
  $4096$ 	 &      $0 $	 & $2.46\times 10^{-9}$ 	 & $4.17\times 10^{-8}$ 	 & $1.93\times 10^{-10}$ 	 & $2.05\times 10^{-9}$\\ \hline   
  $4096$ 	 &      $1 $	 & $7.82\times 10^{-13}$ 	 & $1.35\times 10^{-12}$ 	 & $2.02\times 10^{-13}$ 	 & $5.72\times 10^{-13}$\\ \hline   
 \end{tabular}
  \end{center}
\end{table}



\begin{table}[ht]
  \caption{Relative residual norms:
 GENP with circulant
multipliers filled with $\pm 1$}
  \label{tab65}
  \begin{center}
    \begin{tabular}{| c | c | c | c | c | c | c |}
      \hline
      \bf{dimension} & \bf{iterations} & \bf{mean} & \bf{max} & \bf{min} & \bf{std} \\ \hline
 $256$ & $0$ & $2.37\times 10^{-12}$ & $2.47\times 10^{-10}$ & $9.41\times 10^{-14}$ & $1.06\times 10^{-11}$ \\ \hline
 $256$ & $1$ & $2.88\times 10^{-14}$ & $3.18\times 10^{-12}$ & $1.83\times 10^{-15}$ & $1.36\times 10^{-13}$ \\ \hline
 $512$ & $0$ & $7.42\times 10^{-12}$ & $6.77\times 10^{-10}$ & $3.35\times 10^{-13}$ & $3.04\times 10^{-11}$ \\ \hline
 $512$ & $1$ & $5.22\times 10^{-14}$ & $4.97\times 10^{-12}$ & $3.19\times 10^{-15}$ & $2.29\times 10^{-13}$ \\ \hline
 $1024$ & $0$ & $4.43\times 10^{-11}$ & $1.31\times 10^{-8}$ & $1.28\times 10^{-12}$ & $4.36\times 10^{-10}$ \\ \hline
 $1024$ & $1$ & $1.37\times 10^{-13}$ & $4.33\times 10^{-11}$ & $6.67\times 10^{-15}$ & $1.41\times 10^{-12}$ \\ \hline
  $2048$ 	 &      $0$ 	 & $5.42\times 10^{-9}$ 	 & $1.59\times 10^{-8}$ 	 & $1.54\times 10^{-10 }$	 & $9.04\times 10^{-9}$\\ \hline
  $2048$ 	 &      $1$ 	 & $1.17\times 10^{-13}$ 	 & $2.40\times 10^{-13}$ 	 & $5.14\times 10^{-14}$ 	 & $1.07\times 10^{-13}$\\ \hline
  $4096$ 	 &      $0$ 	 & $1.22\times 10^{-8}$ 	 & $2.47\times 10^{-8}$ 	 & $6.41\times 10^{-10}$ 	 & $1.21\times 10^{-8}$\\ \hline
  $4096$ 	 &      $1$ 	 & $2.29\times 10^{-13}$ 	 & $4.36\times 10^{-13}$ 	 & $1.05\times 10^{-13}$ 	 & $1.81\times 10^{-13}$\\ \hline
   \end{tabular}
  \end{center}
\end{table}

\begin{table}[ht]
  \caption{Relative residual norms:
 GENP for DFT$(n)$ with Gaussian multipliers}
  \label{tab65a}
  \begin{center}
  \begin{tabular}{| c |c|  c | c |  c |c|}
      \hline
  \bf{dimension} & \bf{iterations} & \bf{mean} & \bf{max} & \bf{min} & \bf{std} \\ \hline
 $256$ & $0$ & $2.26\times 10^{-12}$ & $4.23\times 10^{-11}$ & $2.83\times 10^{-13}$ & $4.92\times 10^{-12}$ \\ \hline
 $256$ & $1$ & $1.05\times 10^{-15}$ & $1.26\times 10^{-15}$ & $9.14\times 10^{-16}$ & $6.76\times 10^{-17}$ \\ \hline
 $512$ & $0$ & $1.11\times 10^{-11}$ & $6.23\times 10^{-10}$ & $6.72\times 10^{-13}$ & $6.22\times 10^{-11}$ \\ \hline
 $512$ & $1$ & $1.50\times 10^{-15}$ & $1.69\times 10^{-15}$ & $1.33\times 10^{-15}$ & $6.82\times 10^{-17}$ \\ \hline
 $1024$ & $0$ & $7.57\times 10^{-10}$ & $7.25\times 10^{-8}$ & $1.89\times 10^{-12}$ & $7.25\times 10^{-9}$ \\ \hline
 $1024$ & $1$ & $2.13\times 10^{-15}$ & $2.29\times 10^{-15}$ & $1.96\times 10^{-15}$ & $7.15\times 10^{-17}$ \\ \hline
  $2048 $	 &      $0$ 	 & $2.11\times 10^{-11}$ 	 & $3.05\times 10^{-11}$ 	 & $1.64\times 10^{-11}$ 	 & $8.08\times 10^{-12}$\\ \hline
  $2048 $	 &      $1$ 	 & $1.47\times 10^{-13}$ 	 & $2.73\times 10^{-13}$ 	 & $8.10\times 10^{-14}$ 	 & $1.09\times 10^{-13}$\\ \hline
  $4096 $	 &      $0$ 	 & $1.36\times 10^{-10 }$	 & $3.01\times 10^{-10}$ 	 & $4.52\times 10^{-11}$ 	 & $1.43\times 10^{-10}$\\ \hline
  $4096 $	 &      $1$ 	 & $6.12\times 10^{-13}$ 	 & $9.69\times 10^{-13}$ 	 & $1.91\times 10^{-13}$ 	 & $3.93\times 10^{-13}$\\ \hline
    \end{tabular}
  \end{center}
\end{table}

Next we cover our tests of GENP pre-processed 
with  some multipliers defined by means of  
combining matrices of Section \ref{ssprsml}.    
The test results are represented in Tables \ref{GENPEx0} and \ref{GENPEx}.

In this series of our tests we set 
$n=128$ and
applied GENP  to 
 matrices of 
(\ref{eqtsts})  and six families of
benchmark matrices from \cite{BDHT13},
 pre-processed with multipliers combining the ones 
of following  three basic families.

{\bf Family 1}: The matrices APF of depth 3 (with $d=3$)
and with a (single) random permutation.

{\bf Family 2}: Sparse circulant matrices $C = \Omega^{-1} D(\Omega {\bf v}) \Omega$, 
where the vector  ${\bf v}$ has been filled with zeros,
except for its ten coordinates filled with  $\pm 1$.
Here and hereafter each sign $+$ or $-$ has been assigned with probability 1/2.

{\bf Family 3}:   Sum of two inverse bidiagonal matrices. 
At first their main diagonals have been filled with the integer 101,
 and their first subdiagonals 
have been  filled with $\pm 1$. Then each matrix
have been  multiplied by a  diagonal 
matrix $\diag(\pm 2^{b_i})$, where $b_i$ were random integers
uniformly chosen from 0 to 3.

We combined these three 
basic families of multipliers
and tested GENP on their ten combinations, listed below. 
For each combination we have performed 1000 tests
 and have recorded the average relative error $||A{\bf x} - {\bf b}||/||{\bf b}||$
with matrices $A$ from the seven benchmark families and vectors ${\bf b}$ being standard Gaussian vectors. Here are these ten combinations.

1. $F = I$, $H$ is a  matrix of Family 1.

2. $F = I$, $H$ is a  matrix of Family 3.

3. $F=H$ is a matrix of Family 1. 

4. $F=H$ is a matrix of Family 3.  

5. $F$ is a matrix of Family 1, $H$ is a matrix of Family 3.

6. $F = I$, $H$ is the product of two matrices of Family 1.

7. $F = I$, $H$ is the product of two matrices of Family 2.

8. $F = I$, $H$ is the product of two matrices of Family 3.

9. $F = I$, $H$ is the sum of two matrices of Families 1  and 3.

10. $F = I$, $H$ is the sum of two matrices of Families 2  and 3.

We tested these multipliers for the same linear systems as in our previous tests
in this section 
and for six classes generated from Matlab, by following
their complete description in Matlab and  \cite{BDHT13}. 
Here is the list of these seven test classes.

1. The matrices $A$ of (\ref{eqtsts}).

2. 'circul': circulant matrices whose first row is a standard Gaussian random vector.

3. 'condex': counter-examples to matrix condition number estimators.

4. 'fiedler': symmetric matrices generated with $(i,j)$
and $(j,i)$ elements equal to $c_i - c_j$
 where $c_1,\dots, c_n$ are i.i.d. standard Gaussian variables.

5. 'orthog': orthogonal matrices  with  $(i,j)$ elements 
$\sqrt{\frac{2}{n+1}}\sin{\frac{ij\pi}{n+1}}$.

6. 'randcorr': random $n\times n$ correlation matrices with random eigenvalues from a uniform distribution. (A correlation matrix is a symmetric positive semidefinite matrix with ones on the diagonal.)

7. 'toeppd': $n\times n$ symmetric, positive semi-definite (SPSD) Toeplitz matrices $T$ equal to the sums of $m$ rank-2 SPSD Toeplitz matrices. Specifically,

$$T = w(1)*T(\theta(1)) + ... + w(m)*T(\theta(m))$$
where $\theta(k)$ are i.i.d.
 Gaussian variables  and 
$T(\theta(k))=(cos(2\pi (i-j)\theta(k)))_{i,j=1}^n$.

In our tests, for some pairs of inputs and multipliers, GENP has
produced no meaningful output. In such cases we filled
 the respective entries of 
 Tables \ref{GENPEx0} and \ref{GENPEx} with $\infty$.

 GENP  pre-processed with our 
 multipliers of the 9th combination of 
three basic families, has produced 
accurate outputs without iterative refinement
for all seven benchmark classes of input matrices.
With the other combinations of the three basic families of our multipliers,
this was achieved from 4 to 6 (out of 7)
benchmark input classes.
For comparison,  
the 2-sided pre-processing 
with PRBT-based multipliers of
 \cite{BDHT13}
and  \cite{BBBDD14}
always required iterative refinement.

\begin{table}[h]
\caption{Relative residual norms output by pre-processed GENP 
with no refinement iterations}

\label{GENPEx0}
  \begin{center}

    \begin{tabular}{|c|c|c|c|c|c|c|c|}
    
    \hline
class & 1 & 2 & 3 & 4 & 5 & 6 &7 \\ \hline
1 &2.61e-13 &6.09e-15 &  $\infty$ &2.62e+02 &7.35e-15 &1.38e-12 &3.04e-13 \\ \hline
2 &2.02e+02 &4.34e-14 &5.34e-16 &  $\infty$ &7.35e+02 &5.27e-15 &3.23e-15 \\ \hline
3 &4.34e-13 &8.36e-15 &  $\infty$ &3.03e+02 &1.94e-14 &3.04e-13 &3.21e-13 \\ \hline
4 &1.48e+01 &1.36e-12 &2.39e-16 &1.01e-11 &4.71e+01 &5.09e-15 &5.12e-15 \\ \hline
5 &3.71e-11 &2.21e-14 &  $\infty$ &2.85e+01 &5.83e-10 &2.23e-12 &1.34e-12 \\ \hline
6 &3.33e-13 &9.36e-15 &  $\infty$ &3.66e-05 &7.04e-15 &3.75e-13 &2.11e-13 \\ \hline
7 &7.76e-12 &3.55e-14 &9.91e+01 &7.90e+00 &7.75e+00 &7.11e+00 &1.05e+01\\ \hline 
8 &7.95e+00 &9.55e-14 &7.56e-16 &  $\infty$ &5.74e+03 &6.51e-15 &3.57e-15 \\ \hline
9 &5.36e-13 &1.51e-14 &4.26e-16 &2.24e-11 &3.68e-13 &6.47e-15 &4.92e-15\\ \hline 
10 &3.50e-12 &8.43e-14 &3.43e-13 &2.90e-10 &1.36e+01 &3.53e-13 &1.67e-13\\ \hline 

\end{tabular}
  \end{center}
\end{table}


\begin{table}[h]
\caption{Relative residual norms output by pre-processed GENP 
followed by a single refinement iteration}

\label{GENPEx}
  \begin{center}

    \begin{tabular}{|c|c|c|c|c|c|c|c|}
    
    \hline
class & 1 & 2 & 3 & 4 & 5 & 6 &7 \\ \hline
1 &1.13e-15 &6.90e-17 &  $\infty$ &1.12e+00 &5.23e-17 &2.10e-16 &1.05e-15 \\ \hline
2 &5.07e-04 &7.71e-17 &1.03e-16 &  $\infty$ &4.40e+02 &1.99e-16 &1.19e-15 \\ \hline
3 &1.14e-15 &7.34e-17 &  $\infty$ &5.43e-13 &5.15e-17 &2.24e-16 &1.10e-15 \\ \hline
4 &1.55e-03 &6.19e-17 &1.31e-16 &5.69e-13 &2.69e+02 &2.13e-16 &1.17e-15 \\ \hline
5 &9.80e-16 &6.96e-17 &  $\infty$ &6.75e+01 &5.35e-17 &2.47e-16 &9.84e-16 \\ \hline
6 &1.08e-15 &6.13e-17 &  $\infty$ &6.35e-13 &5.08e-17 &1.86e-16 &1.05e-15 \\ \hline
7 &3.47e+01 &6.17e-17 &2.61e+06 &5.21e+00 &5.31e-17 &1.97e-16 &9.97e-16\\ \hline 
8 &2.56e-04 &6.67e-17 &1.15e-16 &  $\infty$ &7.96e+02 &1.98e-16 &1.08e-15 \\ \hline
9 &9.81e-16 &7.44e-17 &3.99e-17 &6.40e-13 &5.09e-17 &2.02e-16 &1.15e-15\\ \hline 
10 &9.79e-16 &8.32e-17 &1.14e-16 &7.34e-13 &4.07e+01 &2.23e-16 &1.04e-15\\ \hline 

\end{tabular}
  \end{center}
\end{table}


%


Finally we present the results of our tests of GENP with additive pre-processing 
applied to the same
$n\times n$ test matrices $A$ 
of (\ref{eqtsts}), but for $n = 128, 256, 512, 1024$. 
In this case we applied GENP to the matrix $C = A- UV^T$
where $U$ and $V$ were $n\times h$ random Gaussian 
subcirculant matrices, each defined by the $n$ i.i.d.
Gaussian entries of its first column and
scaled so that $||A||=2||UV^T||$. 
Then we computed the solution ${\bf x}$
to the linear system  $A{\bf x}={\bf b}$ for a Gaussian vector ${\bf b}$
by substituting 
the SMW
formula (\ref{eqsmw}) 
into the equation  ${\bf x}=A^{-1}{\bf b}$.

We present the test results in Table \ref{GENPAdd01}. The statistics were taken over 100 runs for each $n$.
 The results changed  little when we 
scaled the matrices $U$ and $V$ to increase the ratio
 $||A||/||UV^T||$ to 10 and 100.

\begin{table}[h]
\caption{Relative residual norms of GENP with Gaussian subcirculant additive pre-processing}
\label{GENPAdd01}
  \begin{center}

    \begin{tabular}{| c | c | c | c | c | c | c |}

\hline

%
%

n 	 & h 	 & iterations 	 & mean 	 & max 	 & min 	 & std \\ \hline
   128 	 &      4 	 &      0 	 & 1.47e-10 	 & 2.51e-09 	 & 1.05e-12 	 & 3.85e-10\\ \hline
   128 	 &      4 	 &      1 	 & 1.58e-14 	 & 3.39e-13 	 & 1.26e-15 	 & 4.48e-14\\ \hline
   256 	 &      4 	 &      0 	 & 8.14e-10 	 & 2.87e-08 	 & 1.20e-11 	 & 3.02e-09\\ \hline
   256 	 &      4 	 &      1 	 & 3.57e-14 	 & 1.16e-12 	 & 2.52e-15 	 & 1.24e-13\\ \hline
   512 	 &      4 	 &      0 	 & 8.86e-09 	 & 3.03e-07 	 & 4.42e-11 	 & 3.44e-08\\ \hline
   512 	 &      4 	 &      1 	 & 2.16e-13 	 & 1.35e-11 	 & 4.60e-15 	 & 1.36e-12\\ \hline
  1024 	 &      4 	 &      0 	 & 2.12e-08 	 & 3.06e-07 	 & 1.45e-10 	 & 4.98e-08\\ \hline
  1024 	 &      4 	 &      1 	 & 9.87e-14 	 & 1.95e-12 	 & 6.64e-15 	 & 2.38e-13\\ \hline

\end{tabular}
  \end{center}
\end{table}


\section{Conclusions}\label{sconc}


Gaussian
elimination with partial pivoting is the workhorse for modern matrix computations,
but it is significantly slowed down by   communication intensive pivoting, both for inputs of  small and large sizes.  Pre-processing with random and fixed multipliers as well as by means of augmentation are efficient alternatives to pivoting according to the results of extensive tests in this paper and a number of previous papers. Our novel insight provides formal support for these empirical observations
and embolden widening the search area 
for efficient pre-processors. We present our initial findings in our search for new classes of such pre-processors
and confirm their efficiency with our numerical tests.

 In \cite{PZ16} and \cite{PZa}   our techniques yield similar results  for
  the fundamental problem of {\em low-rank approximation of a matrix} and for 
the approximation of two singular spaces of a matrix associated with the two sets of its  largest and smallest singular values separated by a gap, respectively.
  
   For the latter subject our substantial new research progress is under way. 
We also plan to test our algorithms on the state-of-the-art computers and to  extend our present results to the computation of Rank Revealing LU Factorization of \cite{P00}.


\medskip

{\bf Acknowledgements:}
 We thank a reviewer for valuable comments  and acknowledge support by NSF Grants CCF 1116736 and
 CCF--1563942
and  PSC CUNY Award 67699-00 45.






\end{document}